\DeclareSymbolFont{AMSb}{U}{msb}{m}{n}
\documentclass[reqno,centertags,11pt,a4paper,noamsfonts]{amsart}
\pdfoutput=1
\usepackage[svgnames]{xcolor}
\usepackage{graphicx}
\usepackage[english]{babel}

\usepackage[utf8]{inputenc}
\usepackage[babel=true,expansion=alltext,protrusion=alltext-nott,final]{microtype}
\usepackage[margin={3cm},marginparwidth={2.5cm}]{geometry}
\usepackage[
  backend=biber,
  style=numeric,
  sorting=none,
  giveninits=true,
  doi=true,
  maxbibnames=6,
  isbn=false,
  url=true,
  eprint=false,
  sorting=nyt
]{biblatex}

\renewbibmacro{in:}{%
  \ifentrytype{article}
    {}
    {\printtext{\bibstring{in}\intitlepunct}}}

\AtEveryBibitem{%
  \clearfield{note}%
}

\renewbibmacro*{doi+eprint+url}{%
  \printfield{doi}%
  \newunit\newblock
  \iffieldundef{doi}
    {\usebibmacro{url+urldate}}
    {}%
}

\usepackage{textcomp}
\usepackage[sb]{libertine}
\usepackage[libertine,bigdelims,vvarbb]{newtxmath}
\useosf
\usepackage[varqu,varl]{inconsolata}
\usepackage[supsfam=LibertinusSerif-Sup,%
       supscaled=1.2,%
       raised=-.13em]{superiors}
\usepackage{mathtools}
\usepackage[T1]{fontenc}
\usepackage{textcomp}
\usepackage{amsthm}
\usepackage{csquotes}
\usepackage[inline]{enumitem}
\numberwithin{equation}{section}
\usepackage[normalem]{ulem}

\usepackage[cal=cm]{mathalpha}
\usepackage{float}  
\usepackage{caption}
\usepackage{subcaption} 
\usepackage{xspace}         
\usepackage[scientific-notation=true]{siunitx}      
\usepackage{pstricks}
\usepackage{tikz}
\usetikzlibrary{positioning} 
\usetikzlibrary{calc}
\usepackage{pgfplots}
\pgfplotsset{width=10cm,compat=1.9}
\usepackage{bm}
\usepackage{sidecap}
\usepackage{graphicx,color}

\usepackage{amsmath}
\DeclareFontFamily{U}{mathx}{}
\DeclareFontShape{U}{mathx}{m}{n}{<-> mathx10}{}
\DeclareSymbolFont{mathx}{U}{mathx}{m}{n}
\DeclareMathAccent{\widehat}{0}{mathx}{"70}
\DeclareMathAccent{\widecheck}{0}{mathx}{"71}

\newcommand{\C}{\mathcal{C}}

\definecolor{light_gray}{gray}{0.75}
\definecolor{lighter_gray}{gray}{0.5}
\colorlet{light_blue}{blue!20}
\definecolor{dark_green}{rgb}{0.0, 0.6, 0.0}
\definecolor{royal_blue}{rgb}{0.0, 0.22, 0.66}
\definecolor{salmon}{rgb}{1.0, 0.55, 0.41}
\definecolor{gold}{rgb}{0.8, 0.63, 0.21}
\definecolor{navy_blue}{rgb}{0.0, 0.0, 0.5}
\definecolor{crimson}{rgb}{0.79, 0.0, 0.09}
\definecolor{amethyst}{rgb}{0.6, 0.4, 0.8}
\definecolor{alizarin}{rgb}{0.82, 0.1, 0.26}
\definecolor{amaranth}{rgb}{0.9, 0.17, 0.31}
\definecolor{azure}{rgb}{0.0, 0.5, 1.0}
\definecolor{canaryyellow}{rgb}{0.82, 0.41, 0.12}
\definecolor{carrotorange}{rgb}{0.8, 0.33, 0.0}
\definecolor{cadmiumgreen}{rgb}{0.0, 0.42, 0.24}
\definecolor{copper}{rgb}{0.72, 0.45, 0.2}
\definecolor{aqua}{rgb}{0.5, 1.0, 0.83}
\definecolor{awesome}{rgb}{1.0, 0.13, 0.32}
\definecolor{candyapplered}{rgb}{1.0, 0.03, 0.0}
\definecolor{caribbeangreen}{rgb}{0.0, 0.8, 0.6}
\definecolor{indigo}{rgb}{0.0, 0.25, 0.42}

\DeclareMathOperator{\weaklystar}{\rightharpoonup\kern-2.2ex ^* \, \,}

\def\XXint#1#2#3{{\setbox0=\hbox{$#1{#2#3}{\int}$ }
\vcenter{\hbox{$#2#3$ }}\kern-.6\wd0}}

\newcommand{\Cp}{\mathrm{Cap}_{\Omega}}

\newcommand{\R}{\mathbb R}
\newcommand{\N}{\mathbb N}

\renewcommand{\C}{\mathbb C}

\newcommand\norm[1]{\lVert #1 \rVert}
\newcommand\bignorm[1]{\big\lVert #1 \big\rVert}

\newcommand\inner[1]{\langle #1 \rangle}

\newcommand{\tr}{{\rm {tr}\,}}

\newcommand\scpr{\boldsymbol{\cdot}}

\newcommand{\mL}{\mathrm{L}}

\renewcommand{\phi}{\varphi}

\newcommand{\mH}{\mathrm{H}}
\newcommand{\mW}{\mathrm{W}}

\theoremstyle{plain}
\newtheorem{theorem}{Theorem}[section]
\newtheorem{proposition}[theorem]{Proposition}

\newtheorem{corollary}[theorem]{Corollary}
\newtheorem{lemma}[theorem]{Lemma}

\newtheorem*{theorem*}{Theorem}

\theoremstyle{definition}
\newtheorem{definition}[theorem]{Definition}
\newtheorem{remark}[theorem]{Remark}
\newtheorem*{remark*}{Remark}

\usepackage[pagebackref=false]{hyperref}
\hypersetup{
	linktoc=all,
	bookmarksdepth=3,
	colorlinks=true,
	linkcolor=Green,citecolor=Orange,urlcolor=DarkBlue,
	pdftitle={A maximal Hohenberg-Kohn theorem via potential theory},
	pdfauthor={Thiago Carvalho Corso}, 
	pdfsubject={},
	pdfkeywords={}} 
\begin{document}
\numberwithin{table}{section}
%%%%%%%%%%%%%%%%%%%%%%%%%%%%%%%%%%%%%%%%%%%%%%%%%%

\title[A maximal HK theorem for non-interacting systems via potential theory]{A maximal Hohenberg-Kohn theorem for non-interacting systems via potential theory}

%%%%%%%%%%%%%%%%%%%%%%%%%%%%%%%%%%%%%%%%%%%%%%%%%%
\author[T.~Carvalho~Corso]{Thiago Carvalho Corso*}
\address[T.~Carvalho Corso]{Institute of Applied Analysis and Numerical Simulation, University of Stuttgart, Pfaffenwaldring 57, 70569 Stuttgart, Germany}
\email{thiago.carvalho-corso@mathematik.uni-stuttgart.de}

%%%%%%%%%%%%%%%%%%%%%%%%%%%%%%%%%%%%%%%%%%%%%%%%%%

% 35J10 - Schrödinger operator, Schrödinger equation
% 35R30 - Inverse problems for PDEs
% 81Q10 Selfadjoint operator theory in quantum theory, including spectral analysis
% 81V74 Fermionic systems in quantum theory
\date{\today}
\thanks{\emph{Funding information}:  DFG -- Project-ID 572811220 and Project-ID 442047500 -- SFB 1481.\\[1ex]
\textcopyright 2026 by the authors. Faithful reproduction of this article, in its entirety, by any means is permitted for noncommercial purposes.}
%%%%%%%%%%%%%%%%%%%%%%%%%%%%%%%%%%%%%%%%%%%%%%%%%%
%\dedicatory{}
%%%%%%%%%%%%%%%%%%%%%%%%%%%%%%%%%%%%%%%%%%%%%%%%%%
\begin{abstract} We show that for Schr\"odinger operators in a connected domain, the Hohenberg-Kohn theorem holds within the class of Laplace form-bounded external potentials if and only if the single-particle density is strictly positive quasi-everywhere. Furthermore, we show that this condition is satisfied for non-interacting Schr\"odinger operators whenever a ground-state exists. Consequently, we establish the Hohenberg-Kohn theorem for non-interacting systems, and thereby the uniqueness of the Kohn-Sham potential, within the maximal class of Laplace form-bounded distributions. The main ingredient to establish these results is a characterization of regular states, whose proof relies on tools from classical potential theory. Moreover, this characterization reveals that, in the continuum setting, the fundamental mechanism underlying the Hohenberg-Kohn theorem is the (quasi)-unique continuation of the density rather than of the many-body wavefunction.
\end{abstract}
%%%%%%%%%%%%%%%%%%%%%%%%%%%%%%%%%%%%%%%%%%%%%%%%%%
\setcounter{tocdepth}{1}
\maketitle
%\tableofcontents
%%%%%%%%%%%%%%%%%%%%%%%%%%%%%%%%%%%%%%%%%%%%%%%%%%
\setcounter{secnumdepth}{2}
%%%%%%%%%%%%%%%%%%%%%%%%%%%%%%%%%%%%%%%%%%%%%%%%%%

\section{Introduction}

\subsection{Motivation} Density Functional Theory (DFT) is the most widely used theoretical framework for electronic structure calculations, playing a central role in chemistry, condensed-matter physics, and materials science. Its conceptual foundation rests on the fundamental observation that all equilibrium properties of many-electron systems can, in principle, be described entirely in terms of the single-particle ground state density. At the core of this observation lies the Hohenberg--Kohn (HK) theorem, which establishes the uniqueness of the external potential generating a given ground state density. Since its original formulation, the HK theorem has profoundly influenced the development of DFT \cite{EE83,PW95,Koh99,Lev10, ED11, Jon15} and continues to motivate mathematical investigations into the density–potential correspondence and the structure of density-based quantum theories \cite{Gar21,Gar22,PTC+23,PvL23,BCL+25}.

From a mathematical perspective, however, the proof of the Hohenberg--Kohn theorem crucially depends on the regularity of the external potential. Indeed, the original argument of Hohenberg and Kohn \cite{HK64} was merely formal, as no rigorous justification was given for the claim that the same wavefunction cannot satisfy two different Schr\"odinger equations. This issue was subsequently clarified by Lieb in his celebrated work \cite[Theorem 3.2]{Lie83}, who observed that the proof can be completed under suitable regularity assumptions on the potential by invoking the strong unique continuation property of the many-body wavefunction. Since then, considerable effort has been devoted to extending the HK theorem to increasingly singular classes of potentials \cite{Zhou12,Lam18,Gar18,Zhou19,Gar20,Gar21,Gar22,LBP20,LLS22}. The strongest result currently available in the continuum setting is due to Garrigue \cite{Gar18,Gar20}, who established uniqueness for external potentials locally in $\mL^p$ for $p>\max\{2d/3,2\}$ (where $d$ is the physical dimension). Nevertheless, this still falls short of the $\mL^{d/2}_{\mathrm{loc}}$ space, which was identified by Lieb \cite{Lie83} as a natural admissible class of external potentials in DFT. Furthermore, these results do not address the uniqueness problem for the considerably larger class of Laplace form-bounded potentials, which was conjectured to hold in \cite{Gar18}. 

The purpose of the present work is therefore to investigate the validity of the Hohenberg--Kohn theorem for the class of Laplace form-bounded potentials. The motivation for considering this larger class of potentials is not merely one of mathematical generality, but rather arises naturally from recent developments in the rigorous foundations of density functional theory for one-dimensional systems ($d=1$). More precisely, building on the work of Sutter \emph{et al.} \cite{SPR+24}, the author and collaborators have recently developed a mathematically rigorous and differentiable formulation of DFT for one-dimensional systems \cite{Cor25,Cor26a,Cor26b,CL25,Cor26c}. Collectively, these works demonstrate that many of the fundamental problems of DFT -- such as pure-state v-representability, the differentiability of the exchange-correlation functional, and the stability of the corresponding inverse problem -- admit complete solutions in one dimension after enlarging the admissible class of external potentials to Laplace form-bounded distributions. %in the spirit originally envisioned by Kohn and Sham \cite{KS65}. In particular, under suitable boundary conditions, these works provide a complete characterization of pure-state $v$-representability for one-dimensional Schrödinger operators with local potentials, establish analyticity of the density-to-potential map, and prove Lipschitz stability of the corresponding inverse problem. A crucial ingredient underlying all of these results is the enlargement of the admissible class of external potentials from locally integrable functions to the class of Laplace form-bounded distributions. 
Consequently, they strongly suggest that distributional potentials are not merely a technical generalization, but rather constitute the natural functional-analytic setting for a differentiable formulation of DFT. 

In this context, the present work takes a first step towards extending the aforementioned one-dimensional program to the higher-dimensional setting by establishing a Hohenberg--Kohn theorem for non-interacting systems with Laplace form-bounded external potentials.

\subsection{Main contributions}  More precisely, the main contributions of this paper can be summarized as follows:
\begin{enumerate}[label=(\roman*)]
    
    \item We show that the Hohenberg-Kohn theorem holds within the class of Laplace form-bounded potentials if and only if the ground state density is strictly positive quasi-everywhere.
    \item We show that the ground state density of non-interacting systems with Laplace form-bounded potentials is strictly positive quasi-everywhere. As a consequence, we establish the Hohenberg-Kohn theorem for non-interacting systems, in any dimension, within the class of Laplace form-bounded potentials. In particular, this guarantees the uniqueness of the Kohn-Sham potential. 
    \item We introduce the notion of regular states in the continuum setting and provide a complete characterization of such states via quasi-strict positivity of their densities. This characterization is the main structural result underlying the proof of the Hohenberg–Kohn theorem.
\end{enumerate}
Furthermore, we emphasize that the proofs of these results rely in an essential way on tools from classical potential theory. To the best of the author's knowledge, this is the first application of these methods to foundational questions in density functional theory, and constitutes one of the conceptual novelties of the present work.

\section{Main results}
\label{sec:main results}

\subsection*{Preliminary notation}

Let us now introduce the notation necessary to precisely state our main results. First, we denote by $\Omega\subset \R^d$ an open and connected subset of $\R^d$ for some $d\in \N$. We then denote by $\mL^2(\Omega)$ the space of square integrable functions with respect to the Lebesgue measure on $\Omega$. For $n\in \N$, we use the notation $\mathcal{H}_n$ for the space of spinless (or spin-polarized) $n$-particles fermionic wavefunctions, i.e.,
\begin{align}
    \mathcal{H}_n = \bigwedge^n \mL^2(\Omega). \label{eq:fermioninc space}
\end{align}
Throughout this paper we shall consider Schr\"odinger operators of the form
\begin{align}
    H_n(v,w) = -\Delta + \sum_{j\neq k}^n w(x_j,x_k) + \sum_{j=1}^n v(x_j), \label{eq:Hamiltonian}
\end{align}
acting on the fermionic space $\mathcal{H}_n$. To properly define these operators, we need to impose assumptions on $v$ and $w$, and specify the corresponding form domains. Throughout this paper we work with Dirichlet boundary conditions, i.e., we consider operators with the form domain
\begin{align}
    \mathcal{Q}_n \coloneqq  \mathcal{H}_n \cap \mH^1_0(\Omega^n), \label{eq:Dirichlet form domain}
\end{align}
where $\mH^1_0(\Omega^n)$ is the Sobolev space of $\mL^2$ functions in $\Omega^n$ with $\mL^2$ weak gradient and zero boundary conditions on $\partial \Omega^n$. In the case $n=1$, we simply have $\mathcal{H}_1 = \mL^2(\Omega)$ and $\mathcal{Q}_1 = \mH^1_0(\Omega)$. A few comments on further generalizations to other boundary conditions are given later in Remark~\ref{rem:extensions}.

\subsection*{Potential spaces} In view of the chosen boundary conditions, we introduce the following classes of external potentials:
\begin{align}
    \mathcal{V}_{\rm fb}(\Omega) \coloneqq \{ v\in \mathcal{D}'(\Omega;\R): \quad \mbox{$v$ satisfies the form bound~\eqref{eq:form-bound condition}} \}, 
\end{align}
and
\begin{align}
    \mathcal{V}(\Omega) \coloneqq \{ v\in \mathcal{D}'(\Omega;\R) : \mbox{ $v$ satisfies~\eqref{eq:form-bound condition} for any $\varepsilon>0$ for sufficiently large $C_\varepsilon>0$.}\}.
\end{align}
Here, $\mathcal{D}'(\Omega;\R)$ denotes the space of real-valued distributions in $\Omega$ and the form bound is defined as follows: there exists $\varepsilon>0$ and $C_\varepsilon>0$ such that
\begin{align}
|\inner{v,|\phi|^2}| \leq \varepsilon \norm{\phi}_{\mH^1}^2 + C_\epsilon \norm{\phi}_{\mL^2}^2, \quad \mbox{for any $\phi \in C^\infty_c(\Omega)$,} \label{eq:form-bound condition}
\end{align}
where $\inner{\cdot,\cdot}$ denotes the distributional dual pairing. Throughout the paper, we adopt the convention that $\inner{\cdot,\cdot}$ is linear in the left entry and conjugate-linear in the right one. In particular, we see distributions as conjugate linear functionals. However, we note that, since the density and potential are real valued, this convention plays no role in defining the quadratic form of $H_n(v,w)$ below.

Similarly, we define the space of admissible pairwise interactions as
\begin{align*}
    \mathcal{W}_{\rm fb}(\Omega\times \Omega) = \{w\in \mathcal{D}'(\Omega\times \Omega;\R) : \quad\mbox{$w$ satisfies the form bound~\eqref{eq:form-bound condition} for $\phi \in C^\infty_c(\Omega\times \Omega)$} \}.
\end{align*}
However, we remark that the interaction potential $w$ will be mostly fixed (or $=0$) so the space of interaction potentials will not be that relevant to us here.

For such class of potentials, one can define the quadratic form
\begin{align}
    \inner{\Psi, H_n(v,w) \Psi} \coloneqq \int_{\Omega^n} |\nabla \Psi(x)|^2 \mathrm{d} x + \inner{w, \rho_\Psi^{(2)}} + \inner{v, \rho_\Psi}, \quad \Psi \in \mathcal{Q}_n, \label{eq:n-particle form}
\end{align}
where $\rho_\Psi^{(2)}$ and $\rho_\Psi$ are, respectively,  the pair density
\begin{align}
    \rho_{\Psi}^{(2)}(x,y) = n(n-1) \int_{\Omega^{n-2}} |\Psi(x,y, x_3,...,x_n)|^2 \mathrm{d} x_3... \mathrm{d} x_n, \label{eq:pair density}
\end{align}
and the single-particle density
\begin{align}
    \rho_\Psi(x) = n \int_{\Omega^{n-1}} |\Psi(x,x_2,...,x_n)|^2 \mathrm{d} x_2 ... \mathrm{d} x_n \label{eq:density}
\end{align}
associated to the wavefunction $\Psi$.

It is important to note that, for potentials $v\in \mathcal{V}_{\rm fb}(\Omega)$ and $w\in \mathcal{W}_{\rm fb}(\Omega\times \Omega)$, the quadratic form~\eqref{eq:n-particle form} might not be closed in $\mathcal{H}_n$, and therefore, not necessarily define a self-adjoint operator on this space. Nevertheless, we can still define a ground-state via the minimization of the quadratic form. To be precise, we shall work with the following definition of ground-states.
\begin{definition}[Mixed states, ground-states and ground-state densities] \label{def:mixed-states and ground-states} $\quad$ \\
\begin{enumerate}[label=(\roman*)]
\item ($\mathcal{Q}_n$ states) We say that a bounded operator $\Gamma \in \mathcal{B}(\mathcal{H}_n)$ is a $\mathcal{Q}_n$ state if it satisfies $0 \leq \Gamma \leq 1$, $\tr \Gamma = 1$, and the spectral decomposition $\Gamma = \sum_{j} \lambda_j |\Psi_j\rangle \langle \Psi_j |$ satisfies
\begin{align}
    \Psi_j \in \mathcal{Q}_n \quad \mbox{for all $j\in J$, and}\quad  \sum_{j} \lambda_j \norm{\Psi_j}_{\mH^1}^2  < \infty.\label{eq:H1 mixed state}
\end{align} 
\item (Energy and ground-states) For $v\in \mathcal{V}_{\rm fb}(\Omega), w\in \mathcal{W}_{\rm fb}(\Omega\times \Omega)$, we define the energy of a $\mathcal{Q}_n$ state $\Gamma = \sum_j \lambda_j |\Psi_j \rangle \langle \Psi_j|$ as
\begin{align}
    E_{v,w}(\Gamma) \coloneqq \mathrm{tr}\, H_n(v,w) \Gamma = \sum_{j=1}^\infty \lambda_j \inner{\Psi_j, H_n(v,w) \Psi_j}, \label{eq:energy def}
\end{align}
with the quadratic form defined in~\eqref{eq:n-particle form}. We then say that $\Gamma$ is a ground-state of $H_n(v,w)$ if $\Gamma$ minimizes the energy $E_{v,w}(\Gamma)$ over all $\mathcal{Q}_n$ states. 
\item (Ground-state densities) We say that a function $\rho: \Omega \rightarrow \R_+$ is a ground-state density of $H_n(v,w)$ if and only if there exists a ground-state $\Gamma$ of $H_n(v,w)$ with
\begin{align}
    \rho_\Gamma(x) = \sum_{j} \lambda_j \rho_{\Psi_j}(x) = \rho(x) ,\quad \mbox{for almost-everywhere $x\in \Omega$.}\label{eq:density mixed states}
\end{align}
\end{enumerate}
\end{definition}

For convenience, we shall use the hat lowercase notation $\hat{v}$ for the symmetric single-particle operator induced by the dual pairing
\begin{align*}
    \inner{\hat{v} \psi, \phi} = \inner{\psi, \hat{v} \phi} = \inner{v, \overline{\psi} \phi}, \quad \psi, \phi \in C^\infty_c(\Omega),
\end{align*}
and the hat uppercase notation $\hat{V}$ and $\hat{W}$ for the corresponding $n$-particle operators induced by $v \in \mathcal{V}_{\rm fb}(\Omega)$ and $w\in \mathcal{W}_{\rm fb}(\Omega \times \Omega)$ via the forms 
\begin{align}
\inner{\Psi, \hat{V} \Psi} = \inner{v, \rho_\Psi}, \quad\mbox{and}\quad \inner{\Psi, \hat{W} \Psi} = \inner{w, \rho^{(2)}_\Psi}.  \label{eq:n particle potentials}
\end{align}
In other words, $\hat{V} = \sum_{j=1}^n v(x_j)$ and $\hat{W} = \sum_{j\neq k}^n w(x_j,x_k)$.

\subsection{Necessary and sufficient conditions for the Hohenberg-Kohn theorem}

With the notation set, we can now state our first main result. This result provides necessary and sufficient conditions for the Hohenberg-Kohn theorem to hold in terms of the single-particle density.

\begin{theorem}[Necessary and sufficient conditions for the HK theorem] \label{thm:HK criteria} Let $v\in \mathcal{V}_{\rm fb}(\Omega)$, $w\in \mathcal{W}_{\rm fb}(\Omega\times \Omega)$, and $\rho$ be a ground-state density of $H_n(v,w)$ in the sense of Definition~\ref{def:mixed-states and ground-states}. Then the potential $v$ is the unique modulo additive constants potential in $\mathcal{V}_{\rm fb}(\Omega)$ such that $\rho$ is a ground state density of $H_n(v,w)$ if and only if the precise representative of $\rho$ satisfies
\begin{align}
    \rho(x) > 0 \quad \mbox{quasi-everywhere in $\Omega$.} \label{eq:positive condition}
\end{align}
\end{theorem}

\begin{remark}[On the positivity of the density] \label{rem:positive density} While we postpone the precise definition of precise representative and quasi-everywhere to Section~\ref{sec:preliminaries}, let us emphasize that condition~\eqref{eq:positive condition} is much stronger than strict positivity almost everywhere. Roughly speaking, it says that the zero set of the density cannot have (Hausdorff) codimension less than $2$. Moreover, we note that quasi-everywhere is equivalent to everywhere in one dimension. In particular,~\eqref{eq:positive condition} is a natural analogue of the everywhere strict positivity property that plays a central role in the one-dimensional setting \cite{SPR+24,Cor25,Cor26a,Cor26b}.
\end{remark}

\subsection{The HK theorem for non-interacting systems}

The next result shows that the criterion from Theorem~\ref{thm:HK criteria} is satisfied for non-interacting systems.

\begin{theorem}[Quasi-strict positivity of the ground state density] \label{thm:density UCP} Let $v\in \mathcal{V}_{\rm fb}(\Omega)$,  then the precise representative of any ground-state density $\rho$ of $H_n(v,0)$ satisfies
\begin{align*}
    \rho(x) > 0 \quad \mbox{quasi-everywhere in $\Omega$.}
\end{align*}
\end{theorem}

As an immediate corollary of Theorems~\ref{thm:HK criteria} and~\ref{thm:density UCP} we obtain the following version of the HK theorem for non-interacting systems.

\begin{corollary}[HK theorem for non-interacting systems] \label{cor:HK non-interacting} Let $v,v' \in \mathcal{V}_{\rm fb}(\Omega)$ be such that $H_n(v,0)$ and $H_n(v',0)$ have a common ground-state density. Then $v' - v = c$ for some constant $c\in \R$.
\end{corollary}

\begin{remark}[Different boundary conditions] \label{rem:extensions} 
We believe that analogous results should hold under different boundary conditions in $\Omega$. More precisely, if we consider $n$-particle form domains $\mathcal{Q}_n$ constructed from general (single-particle) closed order ideals of $\mH^1(\Omega)$ \cite{Sto93}, i.e., of the form
\begin{align*}
    \mH^1_\Gamma(\Omega) = \{ \psi \in \mH^1(\Omega): \psi^\ast = 0 \quad \mbox{quasi-everywhere in $\Gamma$}\}
\end{align*}
for some quasi-closed set $\Gamma \subset \overline{\Omega}$. Then the same results should hold with the strict positivity condition lifted to quasi-everywhere on $\overline{\Omega}\setminus \Gamma$. In this case, however, one needs a suitable definition of capacity up to the boundary, which may lead to technical complications and require further assumptions on the regularity of $\partial \Omega$. It would also be interesting, although technically challenging, to consider the case where $\Omega$ is merely quasi-open.
\end{remark}

\subsection{A characterization of regular states} \label{sec:regular states} 

The main ingredient in the proof of Theorem~\ref{thm:HK criteria}, and the core result of the paper, is a characterization of regular states, which may be of independent interest. To state it precisely, let us introduce the following definition. 
\begin{definition}[Regular states] We say that a $\mathcal{Q}_n$ state $\Gamma$ is regular if the following holds:
\begin{align*}
    \hat{V} \Gamma = 0 \quad \mbox{for some Laplace form-bounded $v\in \mathcal{D}'(\Omega)$} \quad \iff \quad v = 0.
\end{align*}
In other words, $\Gamma$ is regular if and only if the map $v \in \mathcal{V}_{\rm fb}(\Omega) \mapsto \hat{V} \Gamma \in \mathcal{B}(\mathcal{Q}_n^\ast,\mathcal{Q}_n^\ast)$ is injective.
\end{definition}

We then have the following result.
\begin{theorem}[Characterization of regular states] \label{thm:regular states} Let $\Gamma$ be a $\mathcal{Q}_n$ state. Then $\Gamma$ is regular if and only if the precise representative $\rho_\Gamma$ is quasi-everywhere strictly positive. Moreover, if $\{\rho_\Gamma = 0\}$ has positive capacity, then we can find $v\in \mathcal{V}(\Omega)\setminus \{0\}$ such that $v\geq 0$ and $\hat{V} \Gamma = 0$. 
\end{theorem}

\begin{remark}[On the sharpness of Theorem~\ref{thm:regular states}]
The following simple examples serve to illustrate the optimality of the necessary and sufficient conditions in Theorem~\ref{thm:regular states}:
\begin{enumerate}[label=(\arabic*)]
\item The first example shows that, even if we consider only bounded multiplicative potentials, the quasi-strict positivity cannot be weakened to almost-everywhere strict positivity. Let $\Omega = (0,1)$, $n=2$, and consider the wavefunction 
\begin{align*}
    \Psi(x,y):= \sqrt{2}\sin(2\pi x)_+ \sin(2\pi y)_- - \sqrt{2}\sin(2\pi x)_- \sin(2\pi y)_+ .
\end{align*}
Then $\Psi$ is a Slater determinant and the density $\rho_\Psi(x) = \sin(2\pi x)_+^2 + \sin(2\pi x)_-^2 = \sin(2\pi x)^2$ vanishes only at the point $x=1/2$. However, the potential $v(x) = \mathbb{1}_{(0,1/2)} - \mathbb{1}_{(1/2,1)}$, where $\mathbb{1}_J$ is the characteristic function on $J$, is bounded, and satisfies $\hat{V} \Psi = 0$.

\item The second example shows that regularity of a state is fundamentally different from unique continuation of the many-body wavefunction. Let $\Psi = \phi_1\wedge... \wedge \phi_n$ be a Slater determinant where $\phi_1(x)>0$ is continuous and positive everywhere and $\phi_2$ vanishes on an open set $U\subset \Omega$. Then the wavefunction $\Psi$ vanishes identically on the open set $U^n \subset \Omega^n$, but $\Psi$ is regular as it satisfies the conditions of Theorem~\ref{thm:regular states}.
\end{enumerate}
\end{remark}

\subsection{Outline of the proofs} \label{sec:outline of the proof}

Let us now briefly describe the main steps and tools used in the proofs of Theorems~\ref{thm:HK criteria},~\ref{thm:density UCP}, and~\ref{thm:regular states}.

First, the proof of Theorem~\ref{thm:HK criteria} is immediate from the usual HK argument (aka the variational principle) and Theorem~\ref{thm:regular states}. The details can be found in Section~\ref{sec:proof HK criteria}. 

The proof of Theorem~\ref{thm:regular states} is more involved and we split it into two parts. In the first part, we show that the quasi-strict positivity of the density is necessary for $\Gamma$ to be a regular state. For this, the main step is to establish the existence of infinitesimally Laplace form-bounded measures concentrated on arbitrary compact sets of positive capacity (cf. Lemma~\ref{lem:KLMN measure}). The proof of this existence result combines important results from potential theory, such as the dual characterization of the Bessel capacity and the Maz'ya characterization of infinitesimally form-bounded measures. The details are presented in Section~\ref{sec:necessary}.

The second part of the proof of Theorem~\ref{thm:regular states} consists in showing that quasi-strict positivity of the density is sufficient for a state to be regular. This part is presented in Section~\ref{sec:sufficient} and relies on the following key steps. First, we use a commutator identity and the equation $\hat{V} \Gamma = 0$ to show that $\hat{v}$ commutes with the single-particle density matrix $\gamma_\Gamma$. This allows us to obtain a (partial) mutual spectral decomposition of $\hat{v}$ and $\gamma_\Gamma$. We then combine this mutual spectral decomposition with the locality of $\hat{v}$ and an important result of Fuglede \cite{Fug71b} on the local connectedness of the quasi-topology to conclude that $\hat{v}$ acts as multiplication by a constant on $\ker \gamma_\Gamma^\perp$. This step crucially relies on the quasi-strict positivity of $\rho_\Gamma^\ast$. In the last step, we identify $v$ with a constant in the distributional sense by applying a dense range lemma for multiplication operators, see Lemma~\ref{lem:dense range}. This lemma is an important ingredient of the proof and relies on the quasi-positivity of the density and some approximation arguments.

The proof of Theorem~\ref{thm:density UCP} is carried out in Section~\ref{sec:positive density}. The main ingredient in this proof is a lemma showing that the ground state of the single-particle operator $h=-\Delta +v$ on $\mH^1_0(\Omega)$ is quasi-everywhere positive, see Lemma~\ref{lem:maximum principle}. The proof of this lemma relies on an adaptation of arguments used by Ponce and collaborators \cite{OP16,BP03} to establish a general version of the maximum principle for Schr\"odinger operators with $\mL^1$-potentials; it essentially relies on testing the ground state equation with a suitably chosen sequence of trial states and applying the Poincare inequality for functions vanishing on a positive capacity set to obtain a contradiction.  The rest of the proof of Theorem~\ref{thm:density UCP} then follows from the simple observation that the ground state of the single-particle operator is a natural orbital of the ground state of the many-particle operator.

\subsection{Discussion} We conclude this section with a brief discussion on the implications of the present results, possible extensions, and connections with the existing literature.

First, we remark that the necessary and sufficient conditions in Theorem~\ref{thm:HK criteria} are particularly appealing from the DFT perspective because they only depend on the density. In particular, they allows us to characterize the set of densities on which the Lieb functional can have at most one sub-gradient within the natural class of form-bounded potentials. Moreover, this characterization is completely independent of the interaction potential. Therefore, this result shows that quasi-strict positivity is a necessary condition for differentiability of the Lieb functional with arbitrary interactions.

Second, to the best of the author's knowledge, all previous proofs of the Hohenberg-Kohn theorem (see \cite{Lie83,Gar18,Lam18,Gar20,Gar21,LBP20,LLS22}) are based on the strong unique continuation property (UCP) for eigenfunctions of $H_n(v,w)$, which is a highly technical result and not known to hold for the Schr\"odinger operators considered in the present work; see e.g. \cite{Geo79,SS80,Wol93,KT01,Gar18,Dav20}. However, Theorem~\ref{thm:HK criteria} and Corollary~\ref{cor:HK non-interacting} show that the UCP is neither necessary nor sufficient for the HK theorem. The failure of sufficiency is due to the large class of potentials considered here. In particular, none of the existing proofs in the literature establish uniqueness of representing potentials within this class. The failure of necessity is perhaps more striking; it reveals that the fundamental mechanism underlying the Hohenberg--Kohn theorem is the unique continuation of the density rather than that of the wavefunction. In particular, it opens the way for new proofs of the HK theorem.

Lastly, let us comment on the results of Theorem~\ref{thm:regular states}. The notion of regular states introduced here can be viewed as a continuum analogue of recent notions of regular states that have appeared in several works in the finite-dimensional setting \cite{PvL23,PvL25,BCL+25,CDvG+26a,CDvG+26b,PLWS26}. In these works, the authors realized that regular states play an important role in the problem of unique $v$-representability in both ground state \cite{PvL23,BCL+25,PLWS26} and time-dependent \cite{PvL25,CDvG+26a,CDvG+26b} (variants of) DFT in the finite-dimensional setting. Therefore, understanding the set of regular states seems like an important step towards addressing the v-representability question in the continuum setting. Theorem~\ref{thm:regular states} makes a decisive contribution in this direction by providing a complete characterization of such states.

\section{Mathematical background on Sobolev spaces and capacities} \label{sec:preliminaries}
In this section, we recall some definitions and collect some well-known results in potential theory that will be used throughout the paper. The proof of these results can be found in several standard references \cite{AH96,AG01,Maz11,Pon16,HM18}. 

\subsection{Definitions of capacity and quasi-continuity}
\label{sec:capacity}

We begin by recalling the definitions of capacity, polar sets and quasi-everywhere. 
\begin{definition}[Capacity and polar sets] \label{def:capacity} The relative capacity in $\Omega$ is defined as follows:
\begin{enumerate}[label=(\roman*)]
\item For any compact $K\subset \Omega$, \begin{align}
    \Cp(K) \coloneqq \inf \{\norm{\phi}_{\mH^1}^2: \phi \in C^\infty_c(\Omega) \quad \phi(x) \geq 1 \quad \mbox{for all $x\in K$.}\}. \label{eq:cap def}
\end{align}
\item For an open set $U\subset \Omega$, we set
\begin{align}
    \Cp(U) = \sup_{K \subset U \text{ compact}} \Cp(K). \label{eq:capacity open}
\end{align}
\item For general sets $E\subset \Omega$ we use the formula
\begin{align}
    \Cp(E) = \inf_{E \subset U \subset  \Omega\text{ open}} \Cp(U). \label{eq:capacity general set}
\end{align} 
\end{enumerate}
We then say that a set $E \subset \Omega$ is polar if and only if $\Cp(E) = 0$. Moreover, we say that a function satisfies $f> 0$ quasi-everywhere (q.e.) in $\Omega$ if the sub-level set $\{x\in \Omega: f(x) \leq 0\}$ is polar. 
\end{definition} 

Next, we recall the definition of quasi-continuity and quasi-open and quasi-closed sets.
\begin{definition}[Quasi-continuity] We say that a function $f:\Omega \rightarrow \C$ is quasi-continuous in $\Omega$ if, for any $\varepsilon>0$ there exists an open set $V\subset \Omega$ such that $\Cp(V) < \varepsilon$ and $f$ is continuous in $\Omega\setminus V$.    
\end{definition}

\begin{definition}[Quasi-open and quasi-closed sets]
We say that a set $U\subset \Omega$ is quasi-open if for any $\varepsilon>0$ there exists $\omega \subset \Omega$ open such that $\Cp(\omega) < \varepsilon$ and $U\cup \omega$ is open. A set $E\subset \Omega$ is called quasi-closed if the complement is quasi-open in $\Omega$.
\end{definition}

The next result is one of the main reasons for introducing the notions of capacity and quasi-continuity; it shows that any Sobolev function has a well-defined quasi-continuous representative. For a proof of this result, see e.g. \cite[Theorem 3.3.29 and Proposition 3.3.33]{HM18}, \cite[Theorem 6.2.1]{AH96}, or \cite[Proposition 8.6]{Pon16}.

\begin{theorem}[Precise representatives of $\mH^1$ functions] \label{thm:quasi-representative} Let $f \in \mH^1_0(\Omega)$, then the following holds:
\begin{enumerate}[label=(\roman*)]
\item The Lebesgue set of $f$, defined as
\begin{align*}
    \mathcal{L}_f = \left\{ x\in \Omega: \exists f^\ast(x) \in \C \mbox{ such that }\lim_{r\downarrow 0} \frac{1}{|B_r(x)|} \int_{B_r(x)} |f(y)-f^\ast(x)| \mathrm{d} y = 0 \right\},
\end{align*}
satisfies $\Cp(\Omega\setminus \mathcal{L}_f) = 0$. The function $x\in \mathcal{L}_f \rightarrow f^\ast(x)$ is called the \emph{precise representative} of $f$. 
\item The function $f^\ast$ (extended arbitrarily to $\Omega\setminus \mathcal{L}_f$) is quasi-continuous in $\Omega$.
\item \label{it:q.e. convergence} For any approximating sequence $f_n \rightarrow f$ in $\mH^1_0(\Omega)$, we can extract a subsequence such that $f_n^\ast(x) \rightarrow f^\ast(x)$ quasi-everywhere in $\Omega$.
\end{enumerate}
\end{theorem}

\begin{remark}[Quasi-continuous representative of single-particle densities] Recalling that $\sqrt{\rho_\Gamma} \in \mH^1_0(\Omega)$ for any $\mathcal{Q}_n$-state (by the Cauchy-Schwarz inequality), Theorem~\ref{thm:quasi-representative} shows that any single-particle density of a $\mathcal{Q}_n$ state has a quasi-continuous representative.
\end{remark}

\subsection{Basic properties} In the next propositions, we recall several elementary properties of the capacity, quasi-continuous functions and quasi-open sets that are used throughout the paper.

We start with some well-known properties of the capacity.
\begin{proposition}[Basic properties of capacity]
\label{prop:basic capacity}
The capacity defined in~\eqref{eq:cap def} satisfies the following properties:
\begin{enumerate}[label=(\roman*)]
\item \label{it:subadditivity} (Countable subadditivity) For any countable sequence $\{F_j\}_{j \in J}$ of subsets $F_j \subset \Omega$ we have 
\begin{align}
    \Cp(\cup_{j\in J} F_j) \leq \sum_{j\in J} \Cp(F_j). \label{eq:countable subadditivity}
\end{align}
\item (Continuity from above for quasi-closed sets) \label{it:decreasing} 
For a sequence $E_n\subset E_{n-1} \subset \Omega$ of quasi-closed sets with $\Cp(E_1)<\infty$, it holds that
\begin{align*}
    \lim_{n\rightarrow \infty} \Cp(E_n) = \Cp(\cap_{n = 1}^\infty E_n).
\end{align*}
\item \label{it:inner regularity} (Inner-regularity of quasi-open and closed sets) Let $E$ be a quasi-closed or quasi-open set, then
\begin{align*}
    \Cp(E) = \sup_{\substack{K\subset E\\ K \text{ compact}}} \Cp(K).
\end{align*}
\end{enumerate}
\end{proposition}

\begin{proof} The subadditivity property is well-known, see, e.g. \cite[Proposition 3.3.9~(4)]{HM18} or \cite[Proposition 1.4]{dMa87} (where $\Cp$ is called metaharmonic capacity). The inner-regularity follows from a general capacitability result of Choquet \cite{Cho59} (or more elementary arguments for quasi-open and quasi-closed sets). The continuity from above follows from the corresponding property for compact sets (cf. \cite[Proposition 3.3.9~(2)]{HM18}) and the fact that any quasi-closed set $E$ with $\Cp(E)<\infty$ is quasi-compact, i.e., for any $\varepsilon>0$, there exists a compact $K\subset E$ with $\Cp(E\setminus K) < \varepsilon$.
\end{proof}

Next, we record some simple properties of quasi-open sets.
\begin{proposition}[Basic properties of quasi-open sets] \label{prop:quasi-open sets} The following holds:
\begin{enumerate}[label=(\roman*)]
%\item \label{it:polar sets} (Polar sets are quasi-open and closed) Any polar set is both quasi-open and quasi-closed. 
\item  \label{it:finite intersection} (Finite intersection) A finite intersection of quasi-open sets is quasi-open. In particular, a finite union of quasi-closed sets is quasi-closed.
\item \label{it:countable union} (Countable union) A countable union of quasi-open sets is quasi-open. In particular, a countable intersection of quasi-closed sets is quasi-closed.
\end{enumerate}
\end{proposition}

\begin{proof} The proof is immediate from the definitions, the subadditivity of the capacity, and the corresponding properties for open/closed sets, see e.g. \cite[Lemma 2.3]{Fug71a}.
\end{proof}

We now collect some basic properties of quasi-continuous functions. 
\begin{proposition}[Basic properties of quasi-continuous functions] \label{prop:basic quasi-continuous}
Let $f,g:\Omega \rightarrow \C$ be quasi-continuous functions in $\Omega$. Then the following holds:
\begin{enumerate}[label=(\roman*)]
\item \label{it:preimage} (Preimage as quasi-sets) For any open $U\subset \C$, the preimage $f^{-1}(U)$ is quasi-open.
\item \label{it:algebra} (Algebra property) The functions $f+g$ and $fg$ are quasi-continuous.
\item \label{it:lattice} (Lattice property) If $f$ and $g$ are real-valued, then $f\wedge g := \min\{f,g\}$ and $f\lor g = \max\{f,g\}$ are quasi-continuous. For complex-valued $f$, $|f|$ is quasi-continuous. 
\end{enumerate}
\end{proposition} 

\begin{proof}
    The proof is immediate from the subadditivity of the capacity, the definitions of quasi-open sets and quasi-continuity, and the corresponding properties for continuous functions.
\end{proof}

We now recall an important result that allow us to pass from almost everywhere to quasi-everywhere in weak inequalities for quasi-continuous functions. 
\begin{proposition}[From almost-everywhere to quasi-everywhere]
\label{prop:from q.e to a.e.} Let $f,g:\Omega \rightarrow \C$ be quasi-continuous functions in $\Omega$, then $f= g$ a.e. if and only if $f = g$ q.e.. In particular, $f\geq g$ a.e. if and only if $f\geq g$ q.e. .
\end{proposition}

\begin{proof}
    
    The equality case is a classical result and we refer, e.g., to \cite[Theorem 6.1.4]{AH96} or \cite[Lemma 3.3.30]{HM18} for the proofs. For the inequality case, we note that $f\geq g$ if and only if $(f-g)_- = \min\{f-g,0\} = 0$. Hence the result follows from the lattice property~\ref{it:lattice} and the equality case.
\end{proof}

\begin{remark*}[Simple counterexample to strict inequalities] The equivalence in Proposition~\ref{prop:from q.e to a.e.} is not true for strict inequalities. Indeed, since polar sets on $\R$ are empty, a simple counterexample is any non-negative continuous function in $\R$ that vanishes on a discrete set of points. In particular, the q.e. strict positivity condition from Theorem~\ref{thm:HK criteria} is (much) more restrictive than a.e. strict positivity.
\end{remark*}

\subsection{The Bessel capacity} Throughout the paper, we shall also use the so-called Bessel (or absolute) capacity, which simply corresponds to the case $\Omega = \R^d$ in $\Cp$. 
\begin{definition}[Bessel capacity] \label{def:Bessel capacity}
For any set $E\subset \R^d$, the Bessel capacity is defined as 
\begin{align*}
    \mathrm{Cap}(E) = \mathrm{Cap}_{\R^d}(E).
\end{align*}
\end{definition}
For this capacity, the following equivalent definition in terms of the Bessel potential will be useful later. A proof can be found in \cite[Section 2.2 and Proposition 2.3.13]{AH96}.

\begin{proposition}[Bessel capacity via Bessel potential] \label{prop:Bessel capacity} For any set $E\subset \R^d$, we have
\begin{align}
    \mathrm{Cap}(E) \coloneqq \inf\{ \norm{f}_{\mL^2}^2 : f\geq 0, \quad\text{and}\quad G_1 \ast f \geq 1 \quad \mbox{on $E$}\}, \label{eq:Bessel capacity}
\end{align}
where $G_1$ is the Bessel potential defined via the Fourier transform 
\begin{align}
    \widehat{G}_1(\xi) = (1+|\xi|^2)^{-\frac12}. \label{eq:Bessel def}
\end{align}
\end{proposition}

The next result is a simple comparison inequality between the absolute and relative capacities.
\begin{proposition}[Comparison between capacities] \label{prop:related capacities} For any $E\subset \Omega$ we have
\begin{align}
    \mathrm{Cap}(E) \leq \Cp(E). \label{eq:upper capacity bound}
\end{align}
On the other hand, for any open set $\Omega' \subset \subset \Omega$, there exists a constant $c = c(\Omega',\Omega)$ such that
\begin{align}
    \Cp(E) \leq  c \mathrm{Cap}(E)
    \quad \mbox{for any $E\subset \Omega'$.} \label{eq:capacity comparison}
\end{align}
\end{proposition}

\begin{proof} As $\mathrm{Cap} = \mathrm{Cap}_{\R^d}$, the first inequality is immediate for compact sets, and extends to open and general sets via~\eqref{eq:capacity open} and~\eqref{eq:capacity general set}. Similarly, it suffices to prove the second inequality for compact sets $E$. For such sets, one can use a cut-off function $\eta\in C^\infty_c(\Omega)$ with $\eta = 1$ on $\Omega'$ together with the simple estimate $\norm{u \eta}_{\mH^1} \lesssim \norm{\eta}_{\mW^{1,\infty}} \norm{u}_{\mH^1}$ to show that $\Cp(E) \lesssim_{\Omega'} \norm{u}_{\mH^1}^2$ for any $u \in C^\infty_c(\R^d)$ with $u\geq 1$ in $E$. Taking the infimum over such $u$ then yields the result.
\end{proof}

Using Proposition~\ref{prop:related capacities}, we can show that the relative and absolute capacity induce the same quasi-topology in $\Omega$. Since we could not locate this precise statement in the literature, we include a short proof.
\begin{proposition}[Unique quasi-topology] \label{prop:same quasi-topology} A set $E\subset \Omega$ is polar, respectively, quasi-open with respect to the relative capacity if and only if it is polar, respectively, quasi-open with respect to the absolute capacity. In particular, a function in $\Omega$ is quasi-continuous with respect to the relative capacity if and only if it is quasi-continuous with respect to the absolute capacity. 
\end{proposition}

\begin{proof} We prove the statement for quasi-open sets. The other assertions follow from analogous arguments. First, from~\eqref{eq:upper capacity bound}, it is clear that any quasi-open set $U$ with respect to $\mathrm{Cap}_\Omega$ is also quasi-open with respect to $\mathrm{Cap}$. 

For the opposite implication, let $U \subset \Omega$ be quasi-open with respect to $\mathrm{Cap}$. Let $\Omega_n$ be a sequence of open sets satisfying $\Omega_n \subset \subset \Omega$ and $\cup_{n} \Omega_n = \Omega$ (e.g. $\Omega_n =\{ x\in \Omega : \mathrm{dist}(x,\partial \Omega) > 1/n, |x| < n\}$). Then the sets $U_n = U \cap \Omega_n$ are quasi-open by Proposition~\ref{prop:quasi-open sets}~\ref{it:finite intersection}. In particular, let $c_n = c(\Omega_n,\Omega)\geq 1$ be the constant in Proposition~\ref{prop:related capacities}, then for any $\varepsilon>0$, we can find $\omega\subset \Omega$ such that $\mathrm{Cap}(\omega)<\varepsilon/c_n$ and $\omega \cup U$ is open. In particular, $\omega \cap \Omega_n$ and $\Omega_n \cap (\omega \cup U) = (\Omega_n \cap \omega) \cup (\Omega_n \cap U)$ are open. Hence, by inequality~\eqref{eq:capacity comparison} and the subadditivity of the capacity (see Proposition~\ref{prop:basic capacity}~\ref{it:subadditivity}), 
\begin{align*}
    \Cp(\omega \cap \Omega_n) \leq c(\Omega_n,\Omega) \mathrm{Cap}(\omega \cap \Omega_n) \leq \varepsilon. 
\end{align*}
As $\varepsilon>0$ was arbitrary, we conclude that $U \cap \Omega_n$ is quasi-open with respect to $\Cp$. Since $\Omega = \cup_{n} \Omega_n$, we have $U = \cup_{n} (U \cap \Omega_n)$; therefore, $U$ is $\Cp$-quasi-open by the countable union property in Proposition~\ref{prop:quasi-open sets}~\ref{it:countable union}. 
\end{proof}

\subsection{Trace-class operators and quadratic forms on the Sobolev space} To end this section, we recall two elementary results about the single-particle density matrix of $\mathcal{Q}_n$ states and the sesquilinear form associated to form-bounded external potentials. 

\begin{proposition}[The single-particle density matrix] \label{prop:single-particle density matrix} Let $\Gamma$ be a $\mathcal{Q}_n$ state. Then there exists a unique trace-class operator $\gamma_\Gamma \in \mathcal{B}(\mathcal{H}_1)$ such that
\begin{align}
    \tr \hat{A}  \Gamma = \tr \, a \gamma_\Gamma, \quad \mbox{for any bounded operator $a\in \mathcal{B}(\mathcal{H}_1)$,} \label{eq:trace duality pairing}
\end{align}
where 
\begin{align*}
    \hat{A} = \sum_{j=1}^n 1 \otimes... \overbrace{a}^{\text{$j$th position}} \otimes ...\otimes 1 \quad \mbox{is the second quantization of $a$.}
\end{align*} 
Moreover, $0 \leq \gamma_\Gamma\leq 1$, $\tr \gamma_\Gamma = n$, and the spectral decomposition $\gamma_\Gamma = \sum_{j \in J}n_j |\phi_j\rangle \langle \phi_j|$ satisfies
\begin{align}
   \phi_j \in \mH^1_0(\Omega), \quad j\in J\quad \mbox{and}\quad \sum_{j\in J} n_j \norm{\phi_j}_{\mH^1}^2 < \infty. \label{eq:finite kinetic energy}
\end{align}
In particular, $\gamma_\Gamma$ extends to a bounded operator in $\mathcal{B}(\mH^{-1}_0(\Omega),\mH^1_0(\Omega))$ via the formula
\begin{align}
    \gamma_\Gamma v = \sum_{j\in J} n_j \phi_j \inner{v, \phi_j}, \quad \mbox{for $v\in \mH^{-1}_0(\Omega)$.} \label{eq:density matrix extension}
\end{align}
\end{proposition}
\begin{proof} This result is well-known, so we briefly mention the main arguments for the sake of completeness. First, the existence of $\gamma_\Gamma$ can be proved by using the duality between compact operators and trace-class operators. The positivity and $\tr \gamma_\Gamma = n$ follows from the duality pairing~\eqref{eq:trace duality pairing}. The properties of the spectral decomposition in~\eqref{eq:finite kinetic energy} follows by using the duality pairing~\eqref{eq:trace duality pairing} with the regularized kinetic energy $(-\varepsilon \Delta + 1)^{-1} (-\Delta)$, where $-\Delta$ is the Dirichlet Laplacian on $\Omega$, passing to the limit $\varepsilon \downarrow 0$, and using the assumption that $\Gamma$ is a $\mathcal{Q}_n$ state. 
\end{proof}

\begin{proposition}[Continuous sesquilinear forms via form-bounded distributions] \label{prop:continuous forms} Let $v\in \mathcal{V}_{\rm fb}(\Omega)$, then the sesquilinear form $\mathbf{v}: C^\infty_c(\Omega) \times C^\infty_c(\Omega) \rightarrow \C$ defined via
\begin{align*}
    \mathbf{v}(\psi, \phi) = \inner{v, \overline{\psi} \phi}, \quad \phi, \psi \in C^\infty_c(\Omega),
\end{align*}
extends uniquely to a continuous form in $\mH^1_0(\Omega) \times \mH^1_0(\Omega)$. Moreover, the $n$-particle sesquilinear form
\begin{align}
    \mathbf{V}(\Psi, \Phi) = \inner{ \Psi, \hat{V} \Phi} = \inner{v, \rho_{\Psi \Phi}}, \quad \Psi, \Phi \in C^\infty_c(\Omega^n) \cap \mathcal{H}_n,
\end{align}
where
\begin{align}
    \rho_{\Psi, \Phi}(x) = n \int_{\Omega^{n-1}} \overline{\Psi(x,x_2,...,x_n)} \Phi(x,x_2,...,x_n) \mathrm{d} x_2... \mathrm{d} x_n,
\end{align}
extends uniquely to a continuous sesquilinear form in $\mathcal{Q}_n$. Furthermore, the analogous result holds for the form induced by an interaction in $\mathcal{W}_{\rm fb}(\Omega \times \Omega)$. 
\end{proposition}

\begin{proof}
    From the polarization identity
    \begin{align}
        4 \mathrm{Re} \, \mathbf{v}(\psi, \phi) = \mathbf{v}\left((\psi+\phi), (\psi+\phi)\right) - \mathbf{v}\left((\psi-\phi), (\psi-\phi)\right),
    \end{align}
 the form-bound~\eqref{eq:form-bound condition}, and the homogeneity of the form, it follows that
    \begin{align*}
        |\mathbf{v}(\psi,\phi)|\leq C\norm{\psi}_{\mH^1} \norm{\phi}_{\mH^1}, \quad \mbox{for any $\psi, \phi \in C^\infty_c(\Omega)$.}
    \end{align*}
    As $C^\infty_c(\Omega)$ is dense in $\mH^1_0(\Omega)$, we can extend the form uniquely by approximation. 

    For the $n$-particle version, we use the continuous extension of $\mathbf{v}$ to define $\mathbf{V}(\Psi, \Psi)$ for any state $\Psi \in \mathcal{Q}_n$. Then from the Hoffmann-Ostenhof inequality
    \begin{align*}
        \norm{\sqrt{\rho_\Psi}}_{\mH^1}^2 \lesssim \norm{\Psi}_{\mH^1}^2, 
    \end{align*}
    and the continuity of $\mathbf{v}$ we have 
    \begin{align*}
        |\mathbf{V}(\Psi, \Psi)| \lesssim \norm{\Psi}_{\mH^1}^2.
    \end{align*}
    Hence, this quadratic form defines a bounded sesquilinear form in $\mathcal{Q}_n\times \mathcal{Q}_n$ by the polarization identity. The uniqueness of the extension follows from the fact that, for any $\sqrt{\rho} \in \mH^1_0(\Omega)$, we can find $\Psi \in C^\infty_c(\Omega^n) \cap \mathcal{H}_n$ such that $\sqrt{\rho_{\Psi_n}} \rightarrow \sqrt{\rho}$ in $\mH^1_0(\Omega)$. For a proof of this fact, see, e.g. \cite[Theorem 1.2 and 1.3]{Lie83}.
\end{proof}

\begin{remark}[From forms to operators] Let us also recall that any continuous sesquilinear form $\mathbf{A}: \mathcal{Q}_n \times \mathcal{Q}_n \rightarrow \C$ defines a continuous operator $\hat{A} \in \mathcal{B}(\mathcal{Q}_n,\mathcal{Q}_n^\ast)$ via $\mathbf{A}(\Psi,\Phi) = \inner{\hat{A}\Phi, \Psi}$, where $\mathcal{Q}_n^\ast$ is the dual of $\mathcal{Q}_n$. In particular, $\hat{v}$ can be seen as a bounded operator in $\mathcal{B}(\mH^1_0(\Omega); \mH^{-1}_0(\Omega))$.
\end{remark}
\section{The Hohenberg-Kohn criterion}
\label{sec:HK criterion}
Our goal in this section is to prove Theorems~\ref{thm:HK criteria} and~\ref{thm:regular states}. In the first two subsections we prove the two implications in Theorem~\ref{thm:regular states}. In the last subsection we show how Theorem~\ref{thm:HK criteria} follows from Theorem~\ref{thm:regular states}. 

\subsection{Necessary conditions} 
\label{sec:necessary}
We start by showing that, if the density of a state $\Gamma$ vanishes on a set of positive capacity, then this state cannot be regular. The key for this is the following existence result, whose proof we postpone to the end of this subsection.

\begin{lemma}[Existence of infinitesimally form-bounded measures supported on non-polar sets] \label{lem:KLMN measure} Let $K\subset \Omega$ be a compact set with positive capacity. Then there exists a non-negative measure $\mu\in \mathcal{M}_+(K)$ such that $0<\mu(K)<\infty$ and for any $\varepsilon>0$ there exists $C_\varepsilon>0$ such that
\begin{align}
    \int_K |\phi(x)|^2 \mathrm{d} \mu(x) \leq \varepsilon \norm{\phi}_{\mH^1}^2 + C_\varepsilon \norm{\phi}_{\mL^2}^2, \label{eq:KLMN measure}
\end{align}
for any $\phi \in C^\infty_c(\Omega)$.
\end{lemma}

We also need the following result, whose short proof we give below.

\begin{lemma}[The quadratic form associated to form-bounded measures]\label{lem:continuous form} Let $\mu \in \mathcal{M}_+(K)$ be a  finite Borel measure with compact support $K \subset \Omega$ and satisfying~\eqref{eq:KLMN measure} for some $\varepsilon>0$. Then the following holds:
\begin{enumerate}[label=(\roman*)]
\item (No charge to polar sets) \label{it:no charge to polar sets} For any polar set $E\subset \Omega$ we have $\mu(E) = 0$. 
\item (Integral against precise representatives) \label{it:vanishing functions} For any $\psi \in \mH^1_0(\Omega)$ satisfying $\psi^\ast = 0$ q.e. in $K$, the continuous extension of the form induced by $\mu$ satisfies
\begin{align}
    \inner{\psi, \mu \phi} = 0 \quad \mbox{for any $\phi \in \mH^1_0(\Omega).$} \label{eq:vanishing form}
\end{align}
\end{enumerate}
\end{lemma}

\begin{proof}
We first prove~\ref{it:no charge to polar sets}. Let $E\subset \Omega$ be a polar set, then by Proposition~\ref{prop:same quasi-topology} and Proposition~\ref{prop:Bessel capacity}, we can find a sequence $g_n \in C^\infty_c(\R^d)$ with $g_n \geq 1$ on $E$ and $\norm{g_n}_{\mH^1} \rightarrow 0$ as $n\rightarrow \infty$. Hence multiplying by a fixed cutoff function with $\eta = 1$ on $K$, we can even assume that $g_n\in C^\infty_c(\Omega)$. Thus,
\begin{align*}
    \mu(E) = \mu(E \cap K) \leq \int_\Omega |g_n(x)|^2\mathrm{d} \mu(x) \overset{\text{by~\eqref{eq:KLMN measure}}}{\lesssim} \norm{g_n}_{\mH^1}^2 \rightarrow 0.
\end{align*}

To prove~\ref{it:vanishing functions}, we note that, by the continuity of the form induced by $\mu$ (see Proposition~\ref{prop:continuous forms}), it suffices to prove~\eqref{eq:vanishing form} for $\phi \in C^\infty_c(\Omega)$. Moreover, it suffices to consider real-valued functions $\psi$ and $\phi$. So let $\psi \in \mH^1_0(\Omega;\R)$ be such that $\psi^\ast = 0$ q.e. in $K$ and pick an $\mH^1$-approximating sequence $\psi_n \in C^\infty_c(\Omega;\R)$. After extracting a subsequence, we can assume by Theorem~\ref{thm:quasi-representative}~\ref{it:q.e. convergence} that $\psi_n(x) \rightarrow 0$ q.e. in $K$, and hence $\mu$-a.e. by~\ref{it:no charge to polar sets}. Since $\psi_n$ is real-valued, the continuity of the truncation map 
\begin{align*}
    \psi \in \mH^1_0(\Omega) \mapsto (\psi \land m) \lor (-m) \in \mH^1_0(\Omega) \quad \mbox{for any $m\in \N$ (see e.g. \cite{MM79})} 
\end{align*}
implies that the sequence $(\psi_n \land m) \lor (-m) $ converges to $(\psi\land m)\lor(-m)$ in $\mH^1$. Moreover, the functions $(\psi_n \land m)\lor (-m)$ are continuous, bounded, and converge to $0$ $\mu$-a.e. as $n\rightarrow \infty$. Therefore, by dominated convergence and the continuity of the form induced by $\mu$ in Proposition~\ref{prop:continuous forms},
\begin{align*}
    \inner{(\psi\land m)\lor (-m), \mu \phi} = \lim_{n\rightarrow \infty} \int_\Omega (\psi_n\land m)\lor(-m)(x) \phi(x)\mathrm{d} \mu(x) = 0. 
\end{align*}
We can now take the limit $m\rightarrow \infty$ and use the continuity of the form to complete the proof.
\end{proof}

We can now prove the only if part in Theorem~\ref{thm:regular states}.
\begin{proof}[Proof of necessary condition in Theorem~\ref{thm:regular states}]
Let $\Gamma$ be such that 
\begin{align*}
    \Cp(\mathcal{Z}_\rho) >0,\quad \mbox{where}\quad \mathcal{Z}_\rho \coloneqq \{x\in \Omega : \sqrt{\rho^\ast}(x) =0 \}, 
\end{align*} 
where $\sqrt{\rho^\ast}$ is the precise representative of the single-particle density of $\Gamma$. Then, since $\sqrt{\rho^\ast}$ is quasi-continuous, the set $\mathcal{Z}_\rho$ is quasi-closed by Proposition~\ref{prop:basic quasi-continuous}~\ref{it:preimage}. Consequently, by Proposition~\ref{prop:basic capacity}~\ref{it:inner regularity}, there exists a compact subset $K\subset \mathcal{Z}_\rho$ such that $\Cp(K)>0$. From Lemma~\ref{lem:KLMN measure} there exists a non-trivial finite Borel measure $\mu \in \mathcal{M}_+(K)$ satisfying the infinitesimal form bound~\eqref{eq:KLMN measure}. 

We now show that $\hat{\mu} \Gamma = 0$. For this, we first note that, from the spectral decomposition $\Gamma = \sum_{j\geq 1} \lambda_j |\Psi_j\rangle \langle \Psi_j|$, we have $\hat{\mu} \Gamma =0$ if and only if $\hat{\mu} \Psi_j = 0$ for all $j$. So let us fix $j$ and prove that $\hat{\mu} \Psi_j = 0$. For this, observe that $\rho = \sum_{j\geq 1} \lambda_j \rho_{\Psi_j}\geq \lambda_j \rho_{\Psi_j}$ a.e.. Since the precise representatives are quasi-continuous (by Theorem~\ref{thm:quasi-representative}), we can invoke Proposition~\ref{prop:from q.e to a.e.} to conclude that
\begin{align*}
    \rho^\ast \geq \lambda_j \rho_{\Psi_j}^\ast \mbox{ q.e. on $\Omega$,}\quad \mbox{which implies that } \quad \rho_{\Psi_j}^\ast = 0 \mbox{ q.e. on $K$.} 
\end{align*}
Consequently, by Lemma~\ref{lem:continuous form}~\ref{it:vanishing functions} we have $\inner{\mu, \rho_{\Psi_j}} = 0$. Then, since the $n$-particle quadratic form induced by $\mu$ is non-negative, we can apply Cauchy Schwarz to conclude that
\begin{align*}
    |\inner{\Phi, \hat{\mu} \Psi_j}| \leq \inner{\Phi, \hat{\mu} \Phi}^{\frac12} \inner{\Psi_j, \hat{\mu} \Psi_j}^{\frac12} = \inner{\mu, \rho_\Psi}^{\frac12} \inner{\mu, \rho_{\Psi_j}}^{\frac12} = 0, \quad \mbox{for any $\Phi \in \mathcal{Q}_n$.}
\end{align*}
Thus, $\hat{\mu} \Psi_j = 0$ as desired. 
\end{proof}

For the proof of Lemma~\ref{lem:KLMN measure}, we shall use two additional lemmas. The first one is the dual characterization of the absolute capacity, which is a classical result in potential theory. A proof can be found, e.g., in \cite[Theorem 2.2.7]{AH96}. 

\begin{lemma}[Dual characterization of capacity] \label{lem:duality capacity} The Bessel capacity in Definition~\ref{def:Bessel capacity} satisfies
\begin{align}
    \mathrm{Cap}(K) = \sup_{\mu \in \mathcal{M}_+(K)} \left\{ \frac{\mu(K)}{\norm{G_1 \ast \mu}_{\mL^2(\R^d)}}\right\}^2, \quad \mbox{for any compact $K\subset \R^d$,}\label{eq:dual capacity}
\end{align}
where $\ast$ denotes the standard convolution and $G_1(x)$ is the Bessel potential defined in~\eqref{eq:Bessel def}. Note that $(G_1 \ast \mu)(x)$ is well-defined everywhere (and possibly $+\infty$) for any non-negative measure $\mu$ because the kernel $G_1$ is non-negative. 
\end{lemma}

The second lemma we shall need is an important result of Maz'ya, Verbitsky, and others \cite{Ada76,Maz83,MV95,Ver98,MV05} on the characterization of measures satisfying the (infinitesimal) form-bound with respect to the kinetic energy. Here we use the version in \cite[Theorem 4.1]{MV05}.

\begin{lemma}[A characterization of infinitesimally form-bounded measures \cite{Maz11,MV05}] \label{lem:mazya}
Let $\mu \in \mathcal{M}_+(\R^d)$ be a locally finite Borel measure. Then the inequality
\begin{align}
    \int_{\R^d} |\phi|^2 \mathrm{d} \mu(x) \leq \varepsilon \int_{\R^d} |\nabla \phi(x)|^2 \mathrm{d}x + C_\varepsilon \norm{\phi}_{\mL^2}^2, \quad \mbox{for any $\phi \in C^\infty_c(\R^d)$} \label{eq:form-bound measure}
\end{align}
holds for any $\varepsilon>0$ with large enough $C_\varepsilon>0$ if and only if
\begin{align}
    \lim_{\delta \downarrow 0} \sup\left\{ \frac{\mu(F)}{\mathrm{Cap}(F)} : F\subset \R^d \quad \mbox{compact and  }\mathrm{diam}(F) < \delta \right\} = 0,\label{eq:Mazya criteria}
\end{align}
where we understand the quotient as $0$ whenever both $\mathrm{Cap}(F) = 0$ and $\mu(F) =0$.
\end{lemma}

We can now proceed with the proof of Lemma~\ref{lem:KLMN measure}.

\begin{proof}[Proof of Lemma~\ref{lem:KLMN measure}] For $d=1$, we can simply take $\mu = \delta_x$ for some $x\in K$, since the set $\{x\}$ has positive capacity\footnote{Note that $\delta_x$ is precisely the measure used to produce a counterexample to the HK Theorem in \cite{Cor26a} for 1D system with anti-periodic boundary conditions.}. So let us assume that $d\geq 2$. As $K$ is a compact subset of positive capacity, by Lemma~\ref{lem:duality capacity}, we can find $\sigma \in \mathcal{M}_+(K)$ with
\begin{align}
    0<\norm{G_1 \ast \sigma}_{\mL^2}<\infty, \quad \mbox{where $G_1$ is the Bessel potential~\eqref{eq:Bessel def}.} \label{eq:some first existence}
\end{align}
Since $G_1$ is non-negative and $G_2 = G_1\ast G_1$, where $G_2$ is the Bessel potential $\widehat{G}_2(\xi) = (1+|\xi|^2)^{-1}$, equation~\eqref{eq:some first existence} implies that the function $G_2 \ast \sigma$ is $\sigma$-integrable. In particular, the measure $\sigma$ has no atoms (because $G_2(0) = \infty$ for $d\geq 2$) and the function
\begin{align*}
    f_r(x) \coloneqq \int_{B_r(x)} G_2(x-y) \mathrm{d} \sigma(y) 
\end{align*}
decreases monotonically to zero $\sigma$-a.e. as $r\downarrow 0$ by monotone convergence. We can therefore apply Egorov's theorem to find a compact set $L\subset K$ such that $\sigma(L) >0$ and $f_r(x)$ converges uniformly to zero in $L$. The idea is now to set $\mu\coloneqq \sigma \rvert_{L}$ and show that $\mu$ satisfies Maz'ya's criterion~\eqref{eq:Mazya criteria}. Moreover, since the support of $\mu$ is contained in $\Omega$, it suffices to verify~\eqref{eq:Mazya criteria} for compact subsets of $\Omega$.

For this, first note that the measure $\sigma$ is not trivial because $\sigma(L)>0$. Next, from the uniform convergence of $f_r(x)$ in $L = \mathrm{supp}(\mu)$, for any $\varepsilon>0$ we can find $\delta_0>0$ such that 
\begin{align}
	f_\delta(x)\leq \varepsilon \quad \mbox{for any $x\in L$ and $\delta\leq \delta_0$.} \label{eq:pointwiseest}
\end{align}
Moreover, for any compact $F\subset \Omega$ with $\mathrm{diam}(F)<\delta$ we have $F\subset B_\delta(x)$ for any $x\in F$; therefore
\begin{align}
f_\delta(x) \geq \int_{B_\delta(x)} G_2(x-y) \mathrm{d} \mu(y) \geq \int_{F} G_2(x-y)\mathrm{d} \mu(y) =\left( G_2 \ast (\mu\rvert_F)\right)(x), \quad \mbox{for $x\in F\cap L$.} \label{eq:simple upper est}
\end{align}
We now combine estimates~\eqref{eq:pointwiseest} and~\eqref{eq:simple upper est} with the identity $G_2 = G_1 \ast G_1$ to obtain
\begin{align}
 \varepsilon \mu(F) = \int_{F\cap L} \varepsilon \mathrm{d}\mu(y) \overset{\text{\eqref{eq:pointwiseest}}}{\geq} \int_{F\cap L} f_\delta(y) \mathrm{d}\mu(y) \overset{\text{\eqref{eq:simple upper est}}}{\geq} \int_{F\cap L} (G_2\ast \mu\rvert_{F})(x) \mathrm{d} \mu(y) = \norm{G_1\ast(\mu\rvert_{F})}_{\mL^2}^2.  \label{eq:estimate on measure}
\end{align}
We can now use $\nu \coloneqq \mu\rvert_{F}$ as a trial measure in the duality formula from Lemma~\ref{lem:duality capacity}. Precisely, we have
\begin{align*}
\varepsilon \mathrm{Cap}(F) \overset{\text{\eqref{eq:dual capacity}}}{\geq} \varepsilon \frac{\mu(F)^2}{\norm{G_1 \ast (\mu\rvert_F)}_{\mL^2}^2} \overset{\text{\eqref{eq:estimate on measure}}}{\geq}  \mu(F), \quad \mbox{for any compact $F\subset \Omega$ with $\mathrm{diam}(F) < \delta$.}
\end{align*}
Hence, we can take the supremum over all $F$ with $\mathrm{diam}(F)<\delta$ and take the limits $\delta\downarrow 0$ and $\varepsilon \downarrow 0$ to conclude that $\mu$ satisfies~\eqref{eq:Mazya criteria}. Thus, from Lemma~\ref{lem:mazya} we obtain the desired infinitesimal form-bound for $\mu$, which completes the proof.
\end{proof}

\subsection{Sufficient conditions}
\label{sec:sufficient}
We now prove the converse implication in Theorem~\ref{thm:regular states}. To this end, we shall need three auxiliary lemmas.

The first lemma is a dense range criteria for certain multiplication operators in the Sobolev space. To state it precisely, let us introduce the following definition. We say that $\xi = \{\xi_j\}_{j\in \N}$ belongs to $\ell^2\left(\mH^1_0(\Omega)\right)$ if $\xi_j \in \mH^1_0(\Omega)$ for each $j\in \N$ and
\begin{align*}
    \sum_{j} \norm{\xi_j}_{\mH^1}^2 < \infty. 
\end{align*}
For such a sequence, we define the multiplication operator $B_\xi : \ell^2(\mH^1_0(\Omega)) \rightarrow \mW^{1,1}(\Omega)$ as
\begin{align*}
    B_\xi(\zeta) = \sum_{j\geq 1} \overline{\xi}_j \zeta_j .
\end{align*}
Moreover, let us define the domain
\begin{align}
    \mathrm{dom}_\xi \coloneqq \{ \zeta \in \ell^2(\mH^1_0): B_\xi \zeta \in \mH^1_0 \}
\end{align}
and the density
\begin{align}
    \rho_\xi = B_\xi(\xi). 
\end{align}
Then the following holds.
\begin{lemma}[Multiplication operators with dense range] \label{lem:dense range} For any $\xi \in \ell^2(\mH^1_0)$ the density $\rho_\xi$ satisfies
\begin{align*}
    \sqrt{\rho_\xi} \in \mH^1_0(\Omega).
\end{align*}
Moreover, the set $B_\xi(\mathrm{dom}_\xi)$ is dense in $\mH^1_0(\Omega)$ if and only if 
\begin{align}
    \sqrt{\rho_\xi}^\ast > 0 \quad \mbox{quasi-everywhere in $\Omega$.} \label{eq:mutual positive}
\end{align} 
\end{lemma}

\begin{proof} First note that $B_\xi \zeta \in \mW^{1,1}(\Omega)$ for any $\zeta \in \ell^2(\mH^1_0)$. Indeed, this follows from the product rule for derivatives, the triangle inequality, and the Cauchy-Schwarz inequality. A similar Cauchy-Schwarz argument shows that $\sqrt{\rho_\xi} \in \mH^1_0(\Omega)$. 

To prove the second statement, let us first assume that~\eqref{eq:mutual positive} does not hold. Since $|\xi_j| \leq \sqrt{\rho_\xi}$ a.e., it follows from Proposition~\ref{prop:from q.e to a.e.} and the subadditivity of the capacity that
\begin{align*}
    \Cp(\cap_{j\geq 1} \{x : \xi_j^\ast(x) = 0 \}) \geq \Cp(\{ x : \rho_\xi(x)^\ast = 0\})  > 0.
\end{align*}
Moreover, by the countable intersection property of closed sets in Proposition~\ref{prop:quasi-open sets}~\ref{it:countable union} and the inner-regularity of quasi-closed sets in~\ref{prop:basic capacity}~\ref{it:inner regularity}, we can find a compact set 
\begin{align*}
    K \subset \cap_{j\geq 1} \{x : \xi_j^\ast(x) = 0\} \quad \mbox{with} \quad \Cp(K)> 0.
\end{align*}
Hence, by Lemma~\ref{lem:KLMN measure}, we can find a non-trivial measure $\mu \in \mathcal{M}_+(\Omega)$ with support in $K$ and satisfying the form bound~\eqref{eq:KLMN measure}.  In particular, the truncated operator satisfies
\begin{align*}
    (B_\xi^m \zeta)^\ast(x) = \sum_{j=1}^m \xi_j^\ast(x) \zeta_j^\ast(x) = 0 \quad \mbox{ quasi-everywhere in $K$ for any $\zeta \in \ell^2(\mH^1_0)$ and $m\in \N$.}
\end{align*} 
Hence, by Lemma~\ref{lem:continuous form}~\ref{it:vanishing functions}, we conclude that
\begin{align*}
   \inner{\mu, B_\xi^m \zeta } = \int_{\Omega} (B_\xi^m \zeta)^\ast(x) \mathrm{d}\mu(x) = 0, \quad \mbox{for any $\zeta \in \ell^2(\mH^1_0)$.}
\end{align*}
Moreover, from the continuity of the form associated to $\mu$ (see Proposition~\ref{prop:continuous forms}) we have
\begin{align*}
    \inner{\mu, B_\xi \zeta} = \lim_{m\rightarrow \infty} \inner{\mu, B_\xi^m \zeta} = 0.
\end{align*}
But since $\mu \eta = \mu$ for any cutoff function $\eta \in C^\infty_c(\Omega)$ with $\eta =1$ on $K$, $\mu$ is a non-trivial element in the dual space $\mH^{-1}_0(\Omega)$ by~\eqref{eq:KLMN measure}. Since $\mu$ annihilates the set $B_\xi (\mathrm{dom}_\xi)$, this set can not be dense in $\mH^1_0(\Omega)$. This shows that condition~\eqref{eq:mutual positive} is necessary. 

To prove that~\eqref{eq:mutual positive} is also sufficient, we shall show that any $C^\infty_c(\Omega)$ can be approximated in $\mH^1$ by functions in $B_\xi(\mathrm{dom}_\xi)$. For this, we fix $g\in C^\infty_c(\Omega)$ and use a few approximation arguments.

\textbf{1st Step: Capacitary cut-off} First, let us define 
\begin{align*}
    E_k \coloneqq \{x\in \Omega: \sqrt{\rho_\xi}^\ast(x) \leq 1/k\} \cap K, \quad \mbox{where $K = \mathrm{supp}(g)$.}
\end{align*}
Then $E_{k+1} \subset E_k$ is a decreasing sequence of quasi-closed sets by Propositions~\ref{prop:basic quasi-continuous}~\ref{it:preimage} and~\ref{prop:quasi-open sets}~\ref{it:finite intersection}. As the zero set $\{x:\sqrt{\rho_\xi}^\ast(x) = 0 \}$ is polar by assumption, we can use the subadditivity of the capacity and the continuity from above in Proposition~\ref{prop:basic capacity}~\ref{it:decreasing} to conclude that
\begin{align}
    \lim_{k\rightarrow \infty} \Cp(E_k) = \Cp(\{x:\sqrt{\rho_\xi}^\ast(x) = 0 \}) = 0. \label{eq:zero capacity limit}
\end{align}
Consequently, from the comparison inequality~\eqref{eq:upper capacity bound} and Proposition~\ref{prop:Bessel capacity}, we can find a sequence $f_k = G_1\ast g_k \in \mH^1(\R^d;\R_+)$ with 
\begin{align}
    f_k\geq 1 \mbox{ on $E_k$} \quad\mbox{and}\quad \norm{f_k}_{\mH^1}^2 = \norm{g_k}_{\mL^2}^2 \leq \mathrm{Cap}(E_k) + 1/k \leq \Cp(E_k)+1/k. \label{eq:f approx}
\end{align}
Moreover, we pick $\eta \in C^\infty_c(\Omega)$ with $\eta = 1$ on $K$ and set
\begin{align}
    \eta_k \coloneqq \eta (f_k\land 1).
\end{align}
Then from~\eqref{eq:f approx} we have
\begin{align}
    \eta(x) - \eta_k(x) = 0 \quad \mbox{for any $x\in E_k$, and } \quad \norm{\eta_k}_{\mH^1}^2\lesssim \Cp(E_k)+1/k.  \label{eq:E_n localizer}
\end{align}
These functions will serve as capacitary cut-off functions near the zero set of $\rho^\ast$ inside $K$.

\textbf{2nd Step: regularization of inverse density.} The idea now is to use a regularization of the functions $(\eta-\eta_k)\overline{\xi_j}/\rho_\xi$ as multipliers. For each $k \in \N$, we pick a function $T_k\in C^1(\R)$ satisfying
\begin{align}
    T_k(0) = 0, \quad T_k(t) = \frac{1}{t}, \quad \mbox{for $t\geq \frac{1}{k}$,}\quad \mbox{and}\quad M_k \coloneqq \sup_{t\in \R} |T_k(t)| + |\dot{T}_k(t)| < \infty. \label{eq:Lipschitz approx}
\end{align}
Hence, $T_k$ is (globally) Lipschitz and standard results (see, e.g., \cite[Section 7.4]{GT01}) show that the composition 
\begin{align*}
    T_k(\sqrt{\rho_\xi})(x) \coloneqq T_k\left(\sqrt{\rho_\xi(x)}\right) \quad \mbox{belongs to $\mH^1_0(\Omega)$.}
\end{align*}
Moreover, from the boundedness of $T_k$, this function belongs to $\mH^1_0(\Omega)\cap \mL^\infty(\Omega)$. As this space is an algebra of functions, the function $(\eta-\eta_k) T_k(\sqrt{\rho_\xi})^2$ also belongs to $\mH^1_0(\Omega) \cap \mL^\infty(\Omega)$.  We now consider the functions
\begin{align*}
    \zeta_{j,k} \coloneqq g(\eta-\eta_k) T_k(\sqrt{\rho_\xi})^2 \overline{\xi}_j.
\end{align*}
Since $\sqrt{\rho_\xi} \geq |\xi_j|$ a.e. and $T_k(\sqrt{\rho_\xi})(x) = 1/\sqrt{\rho_\xi}(x)$ a.e. on $K\setminus E_k$, the product $(\eta - \eta_k)|T_k(\sqrt{\rho_\xi}) \overline{\xi}_j| \leq |\xi_j|/\sqrt{\rho_\xi} \leq 1$ a.e. in the support of $g$. Hence, the product and chain rule for derivatives together with~\eqref{eq:Lipschitz approx} and Cauchy Schwarz yields
\begin{align*}
    \norm{ \zeta_{j,k}}_{\mH^1}^2 \lesssim_k \norm{\xi_j}_{\mH^1}^2 + \int_\Omega \frac{|\xi_j|^2}{\rho_\xi} |\nabla \sqrt{\rho_\xi}|^2,
\end{align*}
with a constant independent of $j$. Since $\rho_\xi = \sum_j |\xi_j|^2$, we conclude that the sequence $\zeta_k \coloneqq \{\zeta_{j,k}\}_{j \geq 1}$ belongs to $\ell^2(\mH^1_0)$.

\textbf{3rd Step: Approximating the target function.} We now note that
\begin{align*}
    B_\xi(\zeta_k) = \sum_{j \geq 1} g(\eta - \eta_k) T_k(\sqrt{\rho_\xi})^2 |\xi_j|^2 = g(\eta - \eta_k) T_k(\sqrt{\rho_\xi})^2 \sum_{j\geq 1} |\xi_j|^2 = g(1 - \eta_k),
\end{align*}
because $T_k(\sqrt{\rho_\xi}) = 1/\sqrt{\rho_k}$ a.e. in $x\in K\setminus E_k$, $\eta - \eta_k = 0$ everywhere in $E_k$, and $\eta =1$ on $K= \mathrm{supp}(g)$. Therefore,~\eqref{eq:E_n localizer} and~\eqref{eq:zero capacity limit} implies that
\begin{align*}
    \norm{g-B_\xi \zeta_k}_{\mH^1} = \norm{\eta_k g}_{\mH^1} \lesssim \norm{g}_{\mW^{1,\infty}} \norm{\eta_k}_{\mH^1} \rightarrow 0 \quad \mbox{as $k\rightarrow \infty$,}
\end{align*}
which completes the proof.
\end{proof}

The next lemma we need is the following algebraic identity.
\begin{lemma}[Expectation against the commutator with single-particle excitations] \label{lem:single excitation}
Let $\Gamma$ be a $\mathcal{Q}_n$ state and $v\in \mathcal{V}_{\rm fb}(\Omega)$. Then for any $f,g \in \mH^1_0(\Omega)$ the operator $[\hat{V}, a(g)^\ast a(f)]$, where $a(g)^\ast, a(f)$ are the standard creation and annihilation operators, satisfy
\begin{align}
    \mathrm{Tr}   [\hat{V}, a(g)^\ast a(f)] \Gamma = \inner{\hat{v} g,  \gamma_\Gamma f} - \inner{\hat{v} \gamma_\Gamma g,  f}, \label{eq:commutator id}
\end{align}
where $\gamma_\Gamma$ is the single-particle density matrix of $\Gamma$.
\end{lemma}
\begin{proof} From the definition of the annihilation operator
\begin{align}
    \left(a(f)\Psi\right) (x_1,...,x_{n-1}) = \sqrt{n} \int_{\Omega} \overline{f(x)} \Psi(x,x_1,...,x_{n-1}) \mathrm{d} x \label{eq:annihilation operator}
\end{align}
it follows from Cauchy-Schwarz that
\begin{align*}
    \norm{a(f) \Psi}_{\mH^1(\Omega^{n-1})} \leq \sqrt{n} \norm{f}_{\mL^2} \norm{\Psi}_{\mH^1(\Omega^n)}.
\end{align*}
Hence $a(f)$ maps $\mathcal{Q}_n$ to $\mathcal{Q}_{n-1}$. Similarly, the creation operator 
\begin{align}
    \left(a(g)^\ast \Psi\right)(x_1,..,x_n) = \frac{1}{\sqrt{n}} \sum_{j=1}^n g(x_j) \Psi(x_1, ... \hat{x_j},...,x_n), \label{eq:creation operator}
\end{align}
where $\hat{x_j}$ means that $x_j$ is omitted, satisfy
\begin{align*}
    \norm{a(g)^\ast \Psi}_{\mH^1} \leq \sqrt{n}\norm{\Psi}_{\mH^1(\Omega^{n-1})} \norm{g}_{\mH^1(\Omega)}.
\end{align*}
Hence, the operators $b_{g,f} \coloneqq a(g)^\ast a(f)$ and $b_{f,g} = a(f)^\ast a(g)$ are bounded from $\mathcal{Q}_n$ to $\mathcal{Q}_n$. In particular, the trace of $\Gamma$ against the commutator $[\hat{V}, b_{g,f}]$ is well-defined via
\begin{align}
    \tr\, [\hat{V}, b_{g,f}] \Gamma &= \sum_{j\geq 1} \lambda_j \left(\inner{\hat{V} b_{g,f} \Psi_j,\Psi_j} - \inner{\hat{V} \Psi_j,b_{f,g} \Psi_j}\right) \nonumber \\
    &= \sum_{j\geq 1} \lambda_j \inner{v, \rho_{b_{g,f} \Psi_j,\Psi_j} - \rho_{\Psi_j,b_{f,g} \Psi_j}}, \label{eq:trace commutator id}
\end{align}
where $\Gamma = \sum_{j\geq 1} \lambda_j |\Psi_j\rangle \langle \Psi_j|$ is the spectral decomposition of $\Gamma$. Moreover, from the expressions in~\eqref{eq:annihilation operator} and~\eqref{eq:creation operator} we find
\begin{align*}
    \rho_{b_{g,f}\Psi, \Psi}(x) =& n\overline{g(x)} \int_{\Omega^n} f(y) \overline{\Psi(y,X)} \Psi(x,X) \mathrm{d} X \mathrm{d} y\\
    &+ n\sum_{j=2}^n \int_{\Omega^n} \Psi(x,x_2,...x_n)\overline{g(x_j)} f(y) \overline{\Psi(y,x,...,\hat{x}_j,...,x_n)} \mathrm{d} x_2...\mathrm{d} x_n \mathrm{d} y \\
    =& \overline{g(x)}(\gamma_\Psi f)(x) + n \sum_{j=2}^n \int_{\Omega^n} \Psi(y,x,x_2,..\hat{x}_j,...,x_n)\overline{g(y)} f(x_j)\overline{\Psi(x,x_2,...,x_n)}\mathrm{d} x_2...\mathrm{d} x_n \mathrm{d} y,
\end{align*}
and therefore,
\begin{align*}
    \left(\rho_{b_{g,f} \Psi, \Psi} - \rho_{\Psi, b_{f,g} \Psi}\right)(x) = \overline{g(x)}\left(\gamma_\Psi f\right)(x)- f(x) \left( \overline{\gamma_\Psi g}\right)(x),
\end{align*}
Inserting this identity in~\eqref{eq:trace commutator id} completes the proof.
\end{proof}

The last lemma we shall use is an important result of Fuglede \cite{Fug71a,Fug71b} showing that the quasi-topology is locally connected (see also \cite{AL85} for a more general version for $p$-capacities). 
\begin{lemma}[Quasi-connectedness \cite{Fug71b,AL85}] \label{lem:quasi-connected} Let $\Omega\subset \R^d$ be open and connected. Let $A,B\subset \Omega$ be quasi-open sets satisfying $\Cp(\Omega\setminus (A\cup B)) = 0$ and $\Cp(A\cap B) = 0$. Then either $\Cp(A) = 0$ or $\Cp(B) = 0$. In particular, any quasi-continuous function $f:\Omega \rightarrow \C$ assuming only finitely many values, i.e., $\# f(\Omega) < \infty$ is quasi-everywhere constant.
\end{lemma}

\begin{proof}
    The first statement can be found in \cite[Corollary 1]{AL85}. Note that this reference uses the Bessel capacity instead of the relative one, but since both capacities define the same polar and quasi-open sets in $\Omega$ by Proposition~\ref{prop:same quasi-topology}, the result still holds as stated. For the last statement suppose that $f(\Omega) = \{\lambda_1,..., \lambda_m\}$ and let $O_j\subset \C$ be open sets such that $O_j \cap f(\Omega) = \{\lambda_j\}$. Then the sets $f^{-1}(\{\lambda_j\}) = f^{-1}(O_j)$ are quasi-open by Proposition~\ref{prop:basic quasi-continuous}~\ref{it:preimage}, disjoint, and their union cover $\Omega$. Thus all but one of them is polar and $f$ is q.e. constant.
\end{proof}

We can now complete the proof of Theorem~\ref{thm:regular states}.

\begin{proof}[Proof of the sufficient conditions from Theorem~\ref{thm:regular states}] Suppose that $\Gamma$ is a $\mathcal{Q}_n$ state with $\sqrt{\rho_\Gamma}^\ast >0$ q.e.. Let $v$ be a Laplace form-bounded distribution such that
\begin{align}
    \hat{V} \Gamma = 0. \label{eq:HK eq}
\end{align}
Then our goal is to show that $v = 0$. For this, we proceed in three steps.

\textbf{1st Step: mutual spectral decomposition of $\hat{v}$ and $\gamma_\Gamma$.} Let $\gamma_\Gamma$ be the single-particle density matrix of $\Gamma$ and define the occupation spaces 
\begin{align}
    E_\lambda \coloneqq \ker (\gamma_\Gamma - \lambda) \quad \mbox{for $\lambda \in (0,1]$.}
\end{align} 
Then, from Lemma~\ref{lem:single excitation} and~\eqref{eq:HK eq}, we have
\begin{align*}
    \inner{\hat{v} g, \gamma_\Gamma \phi} = \lambda \inner{\hat{v} g, \phi} = \inner{ \hat{v} \gamma_\Gamma g,  \phi}, \quad \mbox{for any $g\in \mH^1_0(\Omega)$ and $\phi \in E_\lambda$.}
\end{align*}
This implies that $\gamma_\Gamma \hat{v} \phi = \lambda \hat{v} \phi$ for any $\phi \in E_\lambda$, where $\gamma_\Gamma$ is the extension to $\mathcal{B}(\mH^{-1}_0, \mH^1_0)$ in Proposition~\ref{prop:single-particle density matrix}. In particular, $\hat{v}(E_\lambda) \subset \ker (\gamma_\Gamma - \lambda) = E_\lambda$. As $\gamma_\Gamma$ is compact, every $E_\lambda$ is finite dimensional. Thus, the restriction $\hat{v}\rvert_{E_\lambda}$ is a self-adjoint operator in $E_\lambda$. We can therefore find an orthonormal basis $\{\phi_j\}$ of $E_\lambda$ such that $\hat{v} \phi_j = a_j \phi_j$. As $\lambda$ is arbitrary, we can find an orthonormal basis $\{\phi_j\}_{j \in J}$ of $\overline{\oplus_{\lambda \in (0,1]} E_\lambda} = \ker \gamma_\Gamma^\perp$ such that
\begin{align}
    \hat{v} \phi_j = a_j \phi_j \quad \mbox{for some $a_j \in \R$.} \label{eq:eigenvalue equation}
\end{align}

\textbf{2nd Step: quasi-connectedness and locality synchronizes the spectral blocks.} We now use the locality of $\hat{v}$ and Lemma~\ref{lem:quasi-connected} to prove that $\hat{v}$ is constant on $\ker \gamma^\perp$. Precisely, by locality,
\begin{align*}
    \inner{\hat{v} \phi_j, \phi_k \eta} = \mathbf{v}(\phi_j, \phi_k \eta) = \inner{v, \overline{\phi_j} \eta \phi_k} = \inner{ \phi_j \overline{\eta}, \hat{v} \phi_k}, \quad \mbox{for any $\eta \in C^\infty_c(\Omega)$.}
\end{align*}
Therefore,
\begin{align*}
    (a_j-a_k) \inner{\phi_j, \eta \phi_k} = \inner{\hat{v} \phi_j, \eta \phi_k}- \inner{\phi_j \overline{\eta}, \hat{v} \phi_k} = 0 \quad \mbox{for any $\eta \in C^\infty_c(\Omega)$.} 
\end{align*}
In particular, $(a_j-a_k) \phi_j \overline{\phi_k} = 0$ a.e. for any $j,k\in J$. By Proposition~\ref{prop:from q.e to a.e.}, we then have
\begin{align}
    (a_j-a_k) \phi_k^\ast(x) \phi_j^\ast(x) = 0 \quad \mbox{q.e. in $\Omega$ for any $j,k \in J$.} \label{eq:mutual exclusion}
\end{align}
We now define the sets
\begin{align}
    U_a \coloneqq \bigcup_{\substack{j \in J \\ a_j = a}} \{ x: \phi_j(x)^\ast \neq 0 \},  \quad \mbox{for every value $a\in \{a_j\}_{j \in J}$,} \label{eq:Ua def}
\end{align}
and make three important observations about these sets. First, by Propositions~\ref{prop:quasi-open sets} and~\ref{prop:basic quasi-continuous}, these sets are quasi-open. Second, since $\rho_\Gamma^\ast(x) >0$ q.e. in $\Omega$ by assumption, their union cover the set $\Omega$ up to a polar set. Third, by~\eqref{eq:mutual exclusion}, the intersection $U_a \cap U_b$ is polar for any $a\neq b$. Therefore, we can apply Lemma~\ref{lem:quasi-connected} to conclude that all but one set $U_c$ is polar. Since each $\phi_j$ is non-trivial, we conclude that $\{a_j\}_{j \in J} = \{c\}$. Therefore,
\begin{align}
    \hat{v} \phi_j = c \phi_j, \quad \mbox{for any $j \in J$.} \label{eq:localization of v}
\end{align}

\textbf{3rd Step: identification as a distribution.} We now show that $v = 0$ as follows. First, define $\xi_j \coloneqq \sqrt{n_j} \phi_j$, where $n_j$ is the occupation number of $\phi_j$ (i.e., $\phi_j \in \ker \gamma_\Gamma - n_j$), and note that
\begin{align*}
    \sum_{j} \norm{\xi_j}_{\mH^1}^2 = \tr (-\Delta + 1) \gamma_\Gamma < \infty \quad \mbox{and}\quad \rho_\xi = \sum_{j\in J} |\xi_j|^2 = \rho_\Gamma. 
\end{align*}
Then by~\eqref{eq:localization of v}, for any $\eta \in C^\infty_c(\Omega)$ we have
\begin{align*}
    \inner{v \eta, B_\xi(\zeta)} = \sum_{j \in J} \inner{v, \overline{\eta} \xi_j \overline{\zeta}_j} = \sum_{j\in J} \inner{\eta \zeta_j, \hat{v} \xi_j} = c \sum_{j\in J} \inner{\eta \zeta_j, \xi_j} =  c \inner{\eta,B_\xi \zeta}. 
\end{align*}
Since $\rho_\xi^\ast >0$ q.e. by assumption, we can use the dense range Lemma~\ref{lem:dense range} to conclude that $v \eta = c \eta$ as a distribution in $\mH^{-1}_0(\Omega)$. As $\eta$ is arbitrary, we then conclude that $v = c$ in $\mathcal{D}'(\Omega)$. In particular, $\hat{V} = n c$, and therefore, $c = 0$ by~\eqref{eq:HK eq}.
\end{proof}

\subsection{Proof of Theorem~\ref{thm:HK criteria}} \label{sec:proof HK criteria}
We now show how Theorem~\ref{thm:regular states} combined with the classical Hohenberg-Kohn variational argument immediately implies Theorem~\ref{thm:HK criteria}.
\begin{proof}[Proof of Theorem~\ref{thm:HK criteria}] Suppose $\Gamma$ is a ground state of $H_n(v,w)$ satisfying~\eqref{eq:positive condition}. Then let $v'\in \mathcal{V}_{\rm fb}(\Omega)$ be such that $H_n(v',w)$ has a ground state $\Gamma'$ with density $\rho_\Gamma' = \rho_\Gamma =: \rho$. Then by the variational principle
\begin{align*}
    \mathrm{tr} H_n(v',w) \Gamma = \mathrm{tr} H_n(v,w) \Gamma + \inner{v'-v,\rho} \leq \mathrm{tr} H_n(v,w)\Gamma' + \inner{v'-v,\rho} = \mathrm{tr} H_n(v',w) \Gamma',
\end{align*}
which shows that $\Gamma$ is a ground state of $H_n(v',w)$. Moreover, note that, although the forms are not necessarily closed, the ground-state equations 
\begin{align*}
    \left(H_n(v,w) - E_{v,w}(\Gamma)\right)\Gamma = 0 \quad \mbox{and}\quad \left(H_n(v',w) - E_{v',w}(\Gamma)\right) \Gamma = 0
\end{align*}
are the Euler-Lagrange equations of minimization of the energy, and therefore still holds. Thus, assuming both ground state energies to be $0$ (which is possible without loss of generality by shifting the potentials by a constant), from the ground state equation we have $(\hat{V}-\hat{V}') \Gamma = 0$.  As $\Gamma$ is a regular state by Theorem~\ref{thm:regular states}, we conclude that $v=v'$.

For the converse direction, suppose that $\rho_\Gamma$ vanishes on a set of positive capacity. Now let $u\in \mathcal{V}(\Omega)\setminus \{0\}$ be as in Theorem~\ref{thm:regular states}. Then it follows that $H_n(v+u,w) \Gamma = H_n(v,w) \Gamma + \hat{U} \Gamma = 0$. As $u\geq 0$, this implies that $\Gamma$ is a ground state of $H_n(v+u,w)$. Moreover, as $u\neq 0$, there exists a state $\tilde{\Gamma}$ such that $\inner{u, \rho_{\tilde{\Gamma}}} > 0$. In particular, $\mathrm{tr} H_n(v+u,w) \tilde{\Gamma} = \mathrm{tr} H_n(v,w) \tilde{\Gamma} + \inner{u, \rho_{\tilde{\Gamma}}} > \mathrm{tr} H_n(v,w) \tilde{\Gamma}$. Therefore $H_n(v+u,w)$ is a different operator than $H_n(v,w)$ but they have the same ground state density. This shows that the HK theorem does not hold and concludes the proof.
\end{proof}
\section{Quasi-positivity of the density for non-interacting systems}
\label{sec:positive density}

We now turn to the proof of Theorem~\ref{thm:density UCP}. The key step in this proof is the following lemma, whose proof relies on an adaptation of an argument in Orsina and Ponce \cite[Proposition 2.1]{OP16} (see also \cite{BP03}).

\begin{lemma}[Maximum principle for Schr\"odinger operators with form-bounded potentials] \label{lem:maximum principle}
Let $v\in \mathcal{D}'(\Omega;\R)$ be such that 
\begin{align}
    |\inner{v,|\phi|^2}| \leq a \norm{\nabla \phi}_{\mL^2}^2 + C\norm{\phi}_{\mL^2}^2, \quad \mbox{for any $\phi \in C^\infty_c(\Omega)$ and some $a>0$ and $C\geq 0$.} \label{eq:a<1 bound}
\end{align}
Suppose that there exists a minimizer $\psi$ of
\begin{align}
    \lambda(v) \coloneqq \inf_{\psi \in \mH^1_0(\Omega)\setminus \{0\}} \left\{\frac{\int_\Omega |\nabla \psi|^2 + \inner{v,|\psi|^2}}{\norm{\psi}_{\mL^2}^2} \right\}. \label{eq:form h(v)}
\end{align}
Then the minimizer is unique (up to a multiplicative constant) and satisfies
\begin{align}
    \Cp\left(\{x\in \Omega : \psi^\ast(x) = 0\}\right) = 0,
\end{align}
where $\psi^\ast$ is the precise representative of $\psi$.
\end{lemma}

For the proof of this lemma, we use the following version of the Poincare inequality. As the proof was omitted in \cite{BP03}, we sketch the short argument below.

\begin{lemma}[Poincare inequality] \label{lem:poincare}
Let $\omega \subset \R^d$ be an open bounded and connected Lipschitz domain and suppose that $E\subset \omega$ is a compact subset with positive capacity. Then there exists a constant $C= C(E,\omega)$ such that
\begin{align}
    \int_\omega |u|^2 \mathrm{d} x \leq C\int_\omega |\nabla u|^2 \mathrm{d} x,
\end{align}
for any $u\in \mH^1(\omega)$ with $u^\ast(x) = 0$ q.e. in $E$.
\end{lemma}

\begin{proof} By contradiction, suppose the inequality is not true. Then we can find a sequence $u_n \in \mH^1(\omega)$ with $\norm{u_n}_{\mL^2} = 1$ and $\norm{\nabla u_n}_{\mL^2} \rightarrow 0$. From the compact embedding $\mH^1(\omega)\subset \mL^2(\omega)$ (which holds for Lipschitz bounded domains), we can extract a subsequence converging in $\mL^2$ to some $u$. As $\nabla u_n \rightarrow 0$ in $\mL^2$, the convergence is actually strong in $\mH^1$ and $\nabla u = 0$. As $\omega$ is connected, $u$ is constant. Since $u_n^\ast(x) = 0$ on $E$ and $u_n^\ast\rightarrow u^\ast$ q.e., after extracting a subsequence by Theorem~\ref{thm:quasi-representative}~\ref{it:q.e. convergence}, we conclude that $u=0$, which is not possible because $\norm{u}_{\mL^2}=1$.
\end{proof}

We can now proceed with the proof of Lemma~\ref{lem:maximum principle}.

\begin{proof}[Proof of Lemma~\ref{lem:maximum principle}]  First, we note that by the polarization identity and~\eqref{eq:a<1 bound}, we can extend~\eqref{eq:a<1 bound} to any $\phi\in \mH^1_0(\Omega)$ by approximation, see Lemma~\ref{lem:continuous form}. In particular, the numerator in~\eqref{eq:form h(v)} is well-defined and any minimizer of~\eqref{eq:form h(v)} can be interpreted\footnote{Note that we do not investigate the closability of the quadratic form~\eqref{eq:form h(v)}, and therefore, $h(v)$ may not be a well-defined self-adjoint operator. However, whether the form is closed or not is irrelevant to our proof.} as the ground state of $h(v) = - \Delta +v$. 

Let $\psi$ be a minimizer of~\eqref{eq:form h(v)}. Without loss of generality, we can assume that $\psi$ is real-valued (as otherwise both $\mathrm{Re} \, \psi$ and $\mathrm{Im}\, \psi$ are minimizers). Moreover, assuming that $u = \psi_+ = \max\{\psi, 0\}$ is not trivial (otherwise, take $\psi_-$), the result will follow if we can show that
\begin{align}
    \Cp\left(\{u^\ast(x) = 0 \}\right) = 0. \label{eq:polar zero set}
\end{align}
Indeed, in this case, the ground state $\psi = u$ has a constant sign and there cannot be two functions with constant sign that are $\mL^2$-orthogonal. Hence, the ground state is simple, i.e., the minimizer of~\eqref{eq:form h(v)} is unique.

To show~\eqref{eq:polar zero set}, we first note that, since $\psi = u- (\psi)_-$ and $|\nabla \psi|^2 = |\nabla u|^2 + |\nabla \psi_-|^2$ and $|\psi|^2 = |u|^2 + |\psi_-|^2$ a.e., it is clear that both $u$ and $\psi_-$ are ground states of $h(v)$. Hence, assuming $\lambda_1(v) = 0$ (which is always possible by a constant shift on $v$), we have the weak ground state equation
\begin{align}
    \int_\Omega \nabla u(x) \scpr \nabla g(x) \mathrm{d} x + \inner{v,u g} = 0,\quad \mbox{for any $g\in \mH^1_0(\Omega)$. } \label{eq:eigenfunction}
\end{align}

We now take inspiration from \cite{BP03,OP16} and consider the trial states
\begin{align}
    g_\delta(x) = \frac{\phi(x)^2}{\delta + u(x)}, \quad \mbox{for $\delta>0$ and $\phi \in C^\infty_c(\Omega;\R)$ with $0\leq \phi \leq 1$.}
\end{align}
Note that $g_\delta \in \mH^1_0(\Omega)$, as it is a product of functions in $\mH^1_0(\Omega) \cap \mL^\infty(\Omega)$. Moreover, we have
\begin{align}
    \int_\Omega \nabla u \scpr \nabla g_\delta &= -\int_{\Omega}\frac{|\nabla u|^2}{(\delta + u)^2} \phi^2 + \int_\Omega \frac{2\phi \nabla u \scpr \nabla \phi}{\delta + u} \mathrm{d} x \nonumber \\
    &\leq -(1-\epsilon) \int_\Omega \frac{|\nabla u|^2 \phi^2}{(\delta + u)^2} + \frac{1}{\epsilon} \int_\Omega |\nabla \phi|^2 \mathrm{d} x,\label{eq:kinetic est}
\end{align}
for any $\varepsilon>0$ by Young's inequality.
We can now estimate the quadratic form of $v$ as follows. Since $v, u$ and $g_\delta$ are real-valued, from the polarization identity, Young's inequality, and~\eqref{eq:a<1 bound} we have
\begin{align*}
    |\inner{v,ug_\delta}| =& |\inner{v, \frac{u\phi}{u+\delta}  \phi}| = \frac14 \left|\inner{ v, \left(\varepsilon^{\frac12}\frac{u\phi}{u+\delta} + \frac{1}{\varepsilon^{\frac12}} \phi\right)^2} - \inner{v, \left(\varepsilon^{\frac12} \frac{u\phi}{u+\delta} - \frac{1}{\varepsilon^{\frac12}} \phi\right)^2} \right|\\
    \leq& \frac{a}{4} \left(  \bignorm{\nabla (\varepsilon^{\frac12}\frac{u \phi}{u+\delta} + \frac{1}{\varepsilon^{\frac12}}\phi)}_{\mL^2}^2 +  \bignorm{\nabla ( \varepsilon^{\frac12}\frac{u \phi}{u+\delta} - \frac{1}{\varepsilon^{\frac12}}\phi)}_{\mL^2}^2\right) \\
    &+ \frac{C}{4} \left( \bignorm{\varepsilon^{\frac12}\frac{u \phi}{u+\delta} + \frac{1}{\varepsilon^{\frac12}}\phi}_{\mL^2}^2 + \bignorm{ \varepsilon^{\frac12}\frac{u \phi}{u+\delta} - \frac{1}{\varepsilon^{\frac12}}\phi}_{\mL^2}^2\right) \\
    =& \frac{a}{2} \left(\varepsilon \bignorm{\nabla \frac{u\phi}{u+\delta}}_{\mL^2}^2 + \frac{1}{\varepsilon} \norm{\nabla \phi}_{\mL^2}^2\right) + \frac{C}{2} \left(\varepsilon\bignorm{\frac{u\phi}{u+\delta}}_{\mL^2}^2 + \frac{1}{\varepsilon} \norm{\phi}_{\mL^2}^2 \right) .
\end{align*}
Thus combining the above estimate with the simple estimate
\begin{align*}
    \int_{\Omega} \left| \nabla \left(\frac{u \phi}{u+\delta}\right)\right|^2 dx \leq  2\int_{\Omega} \frac{\phi^2 |\nabla u|^2}{(u+\delta)^2} +  2\int_{\Omega} |\nabla \phi|^2,
\end{align*}
we obtain
\begin{align}
    |\inner{v, ug_\delta}| \leq a\varepsilon \int_\Omega \frac{\phi^2 |\nabla u|^2}{(u+\delta)^2}  + a(\varepsilon+(2\varepsilon)^{-1}) \norm{\nabla \phi}_{\mL^2}^2 + \frac{C}{2}(\varepsilon + \varepsilon^{-1})  \norm{\phi^2}_{\mL^2}. \label{eq:v est}
\end{align}
Since $a<\infty$, we can pick $\varepsilon>0$ so small that $\left(1-\varepsilon - a\varepsilon \right) >0$. Thus, using~\eqref{eq:kinetic est} and~\eqref{eq:v est} in the eigenfunction equation~\eqref{eq:eigenfunction}, we conclude that
\begin{align}
    \int_\Omega \phi^2 \frac{|\nabla u|^2}{(\delta + u)^2} \lesssim \norm{\phi}_{\mH^1}^2, \label{eq:nice est}
\end{align}
where the implicit constant is independent of $\phi, u$ and $\delta>0$. Thus, assuming that the zero set $\mathcal{Z}_u = \{x: u^\ast(x) = 0\}$ has positive capacity, we can find a compact $K\subset \mathcal{Z}_u$ with positive capacity (by Proposition~\ref{prop:basic capacity}~\ref{it:inner regularity}). Let $K\subset \subset \omega \subset \subset \Omega$ be open bounded connected and Lipschitz\footnote{Notice that we can always find such a (smooth) $\omega$ by using superlevel sets of mollifications of the characteristic function on $K$ together with the regular value theorem and Sard's theorem}. Then, from the Poincare inequality in Lemma~\ref{lem:poincare}, we have
\begin{align}
    C \int_\omega \log(1+u/\delta)^2 \leq \int_\omega |\nabla \log (1+u/\delta)|^2 = \int_\omega \frac{|\nabla u|^2}{(u+\delta)^2}. \label{eq:Poincare est}
\end{align}
Thus, if we pick $\phi \in C^\infty_c(\Omega,[0,1])$ with $\phi = 1$ on $\omega$, we conclude from~\eqref{eq:Poincare est} and~\eqref{eq:nice est} that
\begin{align*}
    \int_\omega \log(1+u/\delta)^2 \lesssim C \norm{\phi}_{\mH^1}^2. 
\end{align*}
As the right-hand side is independent of $\delta>0$, we can take the limit $\delta \downarrow 0$ to conclude that $u = 0$ a.e. in $\omega$. As $\omega$ is arbitrary, we must have $u=0$ which give us a contradiction and completes the proof.
\end{proof}

To complete the proof of Theorem~\ref{thm:density UCP}, we now need the following version of the fermionic bathtub principle. As this result is usually stated for closed forms, which is not necessarily the case of the form in~\eqref{eq:form h(v)}, we recall its proof below.

\begin{lemma}[From many-particle to single-particle ground states] \label{lem:bathtub principle} Let $v\in\mathcal{V}_{\rm fb}(\Omega)$ and $\Gamma$ be a ground-state of $H_n(v,0)$ in the sense of Definition~\ref{def:mixed-states and ground-states}.
Then there exists a minimizer $\phi$ of~\eqref{eq:form h(v)} and this minimizer belongs to $\ker (\gamma_\Gamma - 1)$, where $\gamma_\Gamma$ is the single-particle density matrix of $\Gamma$.
\end{lemma}

\begin{proof}
    First, we note the identity
    \begin{align*}
        \tr H_n(v,0) \Gamma =  \tr\, h(v) \gamma_\Gamma = \sum_{j=1}^\infty n_j \inner{\phi_j, h(v) \phi_j}
    \end{align*}
    where $\gamma_\Gamma = \sum_{j} n_j |\phi_j \rangle \langle \phi_j |$ is the single-particle density matrix of $\Gamma$. Then, assuming $n_j \geq n_{j+1}$ without loss of generality, the Slater determinant $\Phi \coloneqq \phi_1 \wedge... \wedge \phi_n \in \mathcal{Q}_n$ is also a ground-state of $H_n(v,0)$. Moreover, as the minimizer of $h(v)$, if existing, is unique by Lemma~\ref{lem:maximum principle}, it suffices to show that 
    \begin{align}
        \inf \left\{ \frac{\inner{\psi, h(v) \psi}}{\norm{\psi}^2} : \psi \in \mH^1_0(\Omega) \right\} = \inf \left\{ \frac{\inner{\psi, h(v) \psi}}{\norm{\psi}^2} : \psi \in E \coloneqq \mathrm{span} \{ \phi_1,...,\phi_n\} \right\}. \label{eq:equal minimizers}
    \end{align}
    Indeed, in this case, a minimum of the latter exists by compactness, which is also a minimizer for the former. To prove~\eqref{eq:equal minimizers}, first note that, for any $\psi \in E^\perp \cap \mH^1_0(\Omega)$, the function 
    \begin{align*}
        \alpha \in f_j(\alpha) = \frac{\inner{h(v) (\phi_j + \alpha \psi), (\phi_j + \alpha \psi)}}{\norm{\phi_j + \alpha \psi}^2}, \quad j\leq n,
    \end{align*}
    has a minimum at $\alpha = 0$ by the minimality of $\Phi$. In particular, $\dot{f}(0) = 2 \mathrm{Re} \inner{h(v) \phi_j, \psi} = 0$, for any $\psi \in E^\perp \cap \mH^1_0(\Omega)$ and $j\leq n$, which implies that 
    \begin{align}
        \inner{h(v) \phi, \psi} = 0 \quad \mbox{for any $\psi \in E^\perp \cap \mH^1_0(\Omega)$ and $\phi \in E$.} \label{eq:orthogonality via minimization}
    \end{align}  Moreover, we have 
    \begin{align}
        \inf\left\{ \frac{\inner{\psi, h(v) \psi}}{\norm{\psi}^2} : \psi \in E^\perp \cap \mH^1_0(\Omega)\right\} \geq \inf \left\{ \frac{\inner{\phi, h(v) \phi}}{\norm{\phi}^2}: \phi \in E\right\}, \label{eq:minimization outside}
    \end{align}
    as otherwise, we could replace one of the $\phi_j$ in $\Phi$ by an orbital in $E^\perp \cap \mH^1_0(\Omega)$ to obtain a state with lower energy. Therefore, by writing $\psi = \psi_E + \psi_{E^\perp}$ for $\psi \in \mH^1_0(\Omega)$ and using~\eqref{eq:orthogonality via minimization} and~\eqref{eq:minimization outside}, we obtain~\eqref{eq:equal minimizers}, which completes the proof.
\end{proof}
We now prove Theorem~\ref{thm:density UCP}.

\begin{proof}[Proof of Theorem~\ref{thm:density UCP}]
Let $\Gamma$ be a ground-state of $H_n(v,0)$. Then by Lemma~\ref{lem:bathtub principle} the minimizer $\phi_1$ of $h(v) = -\Delta +v$ exists and is a natural orbital of $\gamma_\Gamma$ with occupation number $n_1 = 1$. In particular, since $\rho_\Gamma = \sum_{j\in J} n_j |\phi_j|^2$, we have $\sqrt{\rho_\Gamma} \geq |\phi_1|$ a.e.. By Proposition~\ref{prop:from q.e to a.e.}, the inequality holds for the precise representatives q.e.. The result now follows because $|\phi_1^\ast(x)| >0$ q.e. in $\Omega$ by Lemma~\ref{lem:maximum principle}. 
\end{proof}

%%%%%%%%%%%%%%%%%%%%%%%%%%%%%%%%%%%%%%%%%%%%%%%%%%
\addtocontents{toc}{\protect\setcounter{tocdepth}{-1}}
%%%%%%%%%%%%%%%%%%%%%%%%%%%%%%%%%%%%%%%%%%%%%%%%%%
\section*{Acknowledgements}

The author thanks Matthias Baur for pointing out reference~\cite{HM18}, which was very useful during this work, and Rafael Antonio Lainez Reyes and Markus Penz for some feedback on an earlier draft of this manuscript.

The author is very grateful to Mih\'aly A. Csirik for sharing the arguments provided by ChatGPT 5.6 Sol Ultra to extend Theorem~\ref{thm:regular states} to arbitrary states, which was an important addition to the paper.

T.~C.~Corso acknowledges funding by the \emph{Deutsche Forschungsgemeinschaft} (DFG, German Research Foundation) - Project number 442047500 through the Collaborative Research Center "Sparsity and Singular Structures" (SFB 1481) and Project number 572811220.

%%%%%%%%%%%%%%%%%%%%%%%%%%%%%%%%%%%%%%%%%%%%%%%%%%
\addtocontents{toc}{\protect\setcounter{tocdepth}{2}}
%%%%%%%%%%%%%%%%%%%%%%%%%%%%%%%%%%%%%%%%%%%%%%%%%%
\appendix
%%%%%%%%%%%%%%%%%%%%%%%%%%%%%%%%%%%%%%%%%%%%%%%%%%

\section*{Data availability}
No datasets were generated or analyzed during the current study.

\section*{Competing interests}

The author has no competing interests to declare that are relevant to the content of this article.

%%%%%%%%%%%%%%%%%%%%%%%%%%%%%%%%%%%%%%%%%%%%%%%%%%

\printbibliography
%\input{bibliography}
%%%%%%%%%%%%%%%%%%%%%%%%%%%%%%%%%%%%%%%%%%%%%%%%%%
%%%%%%%%%%%%%%%%%%%%%%%%%%%%%%%%%%%%%%%%%%%%%%%%%%
\end{document}